\batchmode

\newlength{\textlarg}

\documentclass[10pt,english]{article}
\usepackage[T1]{fontenc}
\usepackage[utf8]{inputenc}
\usepackage{geometry}
\usepackage{color,bm,float,mathrsfs,url,amsmath,amsthm,amssymb,graphicx,setspace}
\usepackage[authoryear]{natbib}
\usepackage{MnSymbol,wasysym}
\usepackage{caption}
\usepackage{subcaption}

\DeclareMathOperator*{\argmax}{arg\,max}

\usepackage[unicode=true,pdfusetitle,
 bookmarks=true,bookmarksnumbered=false,bookmarksopen=false,
 breaklinks=false,pdfborder={0 0 0},pdfborderstyle={},backref=false,colorlinks=true]
 {hyperref}
\hypersetup{
 citecolor=blue, linkcolor=blue, urlcolor=blue}

\usepackage{paralist}
  \theoremstyle{definition}
  \newtheorem{assumption}{\protect\assumptionname}
  \theoremstyle{plain}
 \theoremstyle{definition}

\usepackage{multirow,rotating,adjustbox,booktabs,siunitx,times,appendix,comment}

  \providecommand{\assumptionname}{Assumption}

\theoremstyle{definition}

\newtheorem{proposition}{Proposition}

\newcommand{\R}{\mathbb{R}}

\def\tr{\color{black}}

\makeatletter
\renewenvironment{proof}[1][\proofname] {\par\pushQED{\qed}\normalfont\topsep6\p@\@plus6\p@\relax\trivlist\item[\hskip\labelsep\bfseries#1\@addpunct{.}]\ignorespaces}{\popQED\endtrivlist\@endpefalse}
\makeatother

\begin{document}
 
\title{A perturbed utility route choice model\thanks{This research is funded by the European Research Council (ERC) under the European Union's Horizon 2020 research and innovation programme (grant agreement No. 740369). We thank Michel Bierlaire, Emma Frejinger, {\tr the editor Nikolas Geroliminis, as well as three anonymous referees for very useful comments}.}}

\author{Mogens Fosgerau\thanks{University of Copenhagen; \protect\url{mogens.fosgerau@econ.ku.dk}.}
\and Mads Paulsen\thanks{Technical University of Denmark; \protect\url{madsp@dtu.dk}.}
\and Thomas Kjær Rasmussen\thanks{Technical University of Denmark; \protect\url{tkra@dtu.dk}.}}

\date{\today}
\maketitle
\vspace*{-0.75cm}

\begin{abstract}
{\tr We propose a route choice model in which traveler behavior is represented as a utility maximizing assignment of flow across an entire network under a flow conservation constraint}. Substitution between routes depends on how much they overlap. {\tr  The model is estimated considering the full set of route alternatives, and no choice set generation is required.} 
Nevertheless, estimation requires only linear regression and is very fast. Predictions from the model can be computed using convex optimization, and computation is straightforward even for large networks. {\tr We estimate and validate the model using a large dataset comprising 1,337,096 GPS traces of trips in the Greater Copenhagen road network. }   \end{abstract}

\section{Introduction}

Big data that traces individual vehicles through complex traffic networks is now widely available. This opens new possibilities for estimating route choice models that better reflect actual behavior. However, route choice models face the curse of dimensionality, as the number of potential routes through a realistically sized network is extremely large. 

{\tr This paper formulates a route choice model as an instance of a perturbed utility model. In a general perturbed utility model, a consumer chooses a consumption vector $x$ from some budget set $B$ that solves a utility maximization problem of the form $\hat x= \argmax_{x \in B}(a^\intercal x - F(x))$, i.e.  where the utility function is a linear function ``perturbed'' by subtracting a convex function \citep{McFadden2012, Fudenberg2015,Allen2019}. Perturbed utility models are firmly rooted in modern microeconomic theory and can be interpreted as representing a population of agents whose individual behavior is described by one of a wide range of models, where the additive random utility discrete choice model is an example.\footnote{The choice probability vector $p(v)=(p_1(v),\ldots,p_J(v))$ with $v=(v_1,\ldots,v_J)$ of an additive random utility model with utilities $u_j=v_j+\epsilon_j$ can be found as the solution to the perturbed utility maximization problem $p(v) \in \argmax (v^\intercal p - F(p))$, where $F$ is the convex conjugate of the surplus function; this result holds for any joint distribution of random residuals $\epsilon_j$ \citep{Sorensen2019}.} Alternatively, a perturbed utility model can be interpreted at face value to represent the behavior of individual agents who randomize across options \citep{Allen2019}. 

We utilize this framework to formulate a new kind of route choice model. We consider a traveler, represented as choosing their network flow vector $x$ to maximize a certain perturbed utility function. The linear component of utility is the sum across links $e$ of link utility times individual link flow $x_e$. The convex perturbation is specified as a link-length weighted sum across the network links of positive and increasing convex functions $F(x_e)$, each of these being a function of the individual link flow. The perturbation terms induces the traveler to avoid concentrating their flow on few links. The network structure is incorporated through the budget constraint, which requires that flow is conserved at each node in the network. A final important point of the model setup is that we require that $F(0)=F'(0)=0$ while link utilities are negative (representing travel cost with the sign reversed). This has the consequence that link flows will be zero in parts of the network where the flow conservation constraint is not active.

In this paper, we formulate the model, analyze its properties, derive an estimator, and then apply the model to simulated and real data. We will show that the new model has a range of desirable properties. In particular, it allows all physically possible network flows; it predicts that many links are unused by any traveler, while no choice set generation is required. Furthermore, we will show that it implies realistic substitution patterns while being very fast to estimate.}

We use the first-order conditions for the traveler's perturbed utility maximization problem to formulate a linear regression equation, which allows the parameters of the route choice model to be estimated from observed data. We then present a transformation that eliminates the flow conservation constraint such that ordinary least squares (OLS) estimation is applicable. This is a major step forward, as the simplicity of linear regression allows realistic networks to be handled at low computational cost. 

We add to a long line of research into route choice models. Early route choice models have relied on maximum likelihood estimation of an additive random utility discrete choice model \citep{McFadden1981} for the choice between alternative routes. However, the number of possible routes in a large network is extremely large, comparable to the number of atoms in the universe, even if loops are ruled out. Therefore, the main problem for these models is that the number of potential routes is prohibitively large. Much attention has been given to the generation of choice sets with good coverage to avoid bias resulting from excluding relevant alternatives \citep{Prato2009}. Similarly, much attention has been given to finding models that lead to realistic substitution patterns when alternatives routes overlap more or less.\footnote{\citet{Prato2009} provides a review. {\tr More recent overviews of the literature may be found in \citet{Oyama2020} and \citet{Duncan2020}.}} The perturbed utility route choice model (PURC) operates at the network level and does not require a choice set as input. On the contrary, the PURC model predicts which links are active, and the set of routes using these links can be thought of as a consideration set.

Another issue with additive random utility discrete choice models of route choice is that error terms are generally assumed to have full support such that every alternative is chosen with positive probability. This is neither necessarily realistic nor desirable in assignment, as most feasible routes in a large network are quite nonsensical \citep{Watling2015,Rasmussen2017,Watling2018}. {\tr The bounded choice model (BCM) \citep{Watling2018} assigns zero probability to alternatives with a random utility that exceeds an exogenously defined upper bound\footnote{Routes are assigned probabilities based on the distribution of the random utility difference to an ``imaginary'' reference alternative. Probabilities of the routes relate to the odds of choosing each route versus the ’imaginary’ reference alternative. Setting the reference alternative as the route with the minimum generalized cost (maximum systematic utility) is a special case of the BCM, causing routes to only receive non-zero probabilities if their costs are within the bound of the cheapest route. See supplementary material of \citet{Duncan2021} for further details.}. Unlike the PURC model, the BCM requires path enumeration, as the distinction between used and unused alternatives is based on path costs. For large-scale applications with realistic values of the bound on random utility, BCM thus requires the enumeration of very large choice sets.}


A stream of research has considered recursive models in which the traveler is seen as choosing a path link by link in a Markovian fashion. Earlier papers considered the assignment problem \citep{Dial1971, Bell1995, Shen1996a, Baillon2008}.  A recent series of papers has considered estimation by maximum likelihood of what they term the recursive logit model {\tr and generalizations building on the multivariate extreme value distribution  \citep{Fosgerau2013e, Mai2015, Mai2015a, Mai2016}}. In contrast, estimation of the PURC model does not require computationally demanding maximization of a likelihood function. Whereas the basic recursive logit model has the independence of irrelevant alternatives (IIA) property at the route choice level, the PURC model predicts realistic substitution patterns that derive directly from the network structure. {\tr Various route choice models have been proposed that account for substitution patterns across alternatives. The path-size logit (PSL) route choice model \citep{Ben-Akiva1999PSL} adjusts the systematic utility of an alternative based on the overlap with other alternatives in the choice set. Some variations of the PSL model are reviewed in \citet{Duncan2020}. The adaptive path-size logit (APSL) model, proposed by \citet{Duncan2020}, is internally consistent in the sense that the adjustment factors for overlap are calculated based on route choice probabilities. 
Recently, the integration of path-size concepts, including the APSL, into the BCM have been explored by \citet{Duncan2021}. The challenge of path enumeration, however, still remains in these models.  
}

\citet{Oyama2020} cast the assignment problem for the network generalized extreme value model as a concave maximization problem of the perturbed utility form, where the perturbation function is a generalized entropy that incorporates the network structure. This is an instance of the general result that any additive random utility {\tr discrete choice} model can be represented as a perturbed utility model \citep{Hofbauer2002}. In contrast to \citet{Oyama2020}, we here consider estimation of the model parameters. Furthermore, we use a perturbation function that allows corner solutions, which avoids assigning positive probability to all routes that are physically possible.

We set up the model in Section \ref{sec:setup} and derive a linear regression equation that can be estimated by OLS regression. We find that the model implies that routes are substitutes and more so the more they overlap. In Section \ref{sec:experiments}, we illustrate the model's predictions using first a toy network and then a large-scale network covering the Copenhagen metropolitan area. In Section \ref{sec:estsimdata}, we demonstrate that true parameters can be recovered from realistic simulated data from a large network. We go on in Section \ref{sec:application} to estimate the model using a large dataset of GPS traces for trips in the Copenhagen network. We validate the model's predictions against our data. Section \ref{sec:conclusion} concludes the study.

\section{Setup}\label{sec:setup}

A network $(\mathcal{V},\mathcal{E})$ is defined by an incidence matrix $A$, which has a row for each vertex (node) $v\in\mathcal{V}$ and a column for each edge (link) $e\in\mathcal{E}$. As we are considering traffic networks, we will talk about links and nodes. The entries of $A$ are $a_{ve}=-1$ if edge/link $e$ leads out of vertex/node $v$, $a_{ve}=1$ if link $e$ leads into node $v$, and $a_{ve}=0$ otherwise. 

A traveler has unit demand given by the vector $b$, which is a column vector across nodes with $-1$ at the origin node, $1$ at the destination node, and zeros otherwise. There is a traveler for each demand vector $b$. We suppress the dependence on  $b$ in the notation.


{\tr 
The traveler is viewed as solving the constrained utility maximization problem 
\begin{equation}\label{eq:pertutil}
\begin{aligned}
\max_{x \in \mathbb{R}_+^{|\mathcal{E}|}} U\left( x\right)  \qquad \text{s.t. } Ax=b,
\end{aligned}
\end{equation}
where the constraint ensures that flow is conserved at each node in the network. This assumption may be interpreted as saying that if we observe many individual travelers with a given demand $b$, then their combined  network flow vector $x$ is the solution to the problem \eqref{eq:pertutil}.

The utility function $U$ is specified using the perturbed utility form  \citep{McFadden2012, Fudenberg2015,Allen2019}, consisting of a linear term minus a convex perturbation term.\footnote{\tr Subtracting a convex function rather than adding a concave function follows the convention of convex analysis \citep{Rockafellar1970}.} We specify the utility function as a weighted sum over links weighted by the length of each link as follows.  
\begin{equation}\label{eq:U}
    U\left( x\right)  = l^\intercal\left( u\circ x\right) -l^\intercal F(x).
\end{equation}

We will first explain the linear term $l^\intercal\left( u\circ x\right)$. The vector $u=(u_e)_{e\in \mathcal{E}}$ has a component for each link that expresses the utility rate --- that is, the utility per distance unit --- of using that link. We assume that all components of $u$ are negative. Vector $l=(l_e)_{e\in \mathcal{E}}$ comprises the link lengths. 
The vector $u \circ x = (u_e x_e)_{e\in \mathcal{E}}$ is the component-wise product of the vector of link utility rates and the flow vector. Hence, the term $l^\intercal\left( u\circ x\right)$ is the sum across links of the utility rate multiplied by link flow and by link length. 

The second term, the perturbation, is the sum across links, weighted by link length, of the function $F$ applied component-wise to the flow vector, that is, $F(x)=(F(x_e))_{e\in \mathcal{E}}$. We make the following assumption concerning $F$.

\begin{assumption}\label{ass:F}
$F:\mathbb{R}_+\rightarrow \mathbb{R}$ is strictly convex, $F(0)=0$, and $F'(0)=0$. 
\end{assumption}

The perturbation $F(x_e)$ corresponding to each link  is zero and flat at zero flow, which, together with the condition that link utilities are negative, ensures that the optimal flow on any link is zero if the flow conservation constraint is not active there. The combined  perturbation $l^\intercal F(x)$ is strictly convex, as it is a linear combination of strictly convex terms. The presence of the combined perturbation term makes the utility function  \eqref{eq:U} strictly concave, which ensures a unique solution to the utility maximization problem. 

The perturbation term becomes large and negative if much flow is concentrated on a link. This induces the traveler to distribute their flow across more links. In this way, the perturbation term works in much the same way as a congestion term would work, but with the crucial difference that the perturbation is part of the traveler's preferences and depends not on the behavior of other travelers but only on the individual flow vector $x$. 

We also note that the utility function \eqref{eq:U} is a sum  across links of independent terms. The flow conservation constraint creates dependencies whereby changes in utilities on some links induces substitution with flow on other links. 
}

All terms in utility function \eqref{eq:U} are weighted by link length, which makes the model invariant to link splitting. Consider a link $e$ with a contribution to the utility that is $l_e u_e x_e - l_e F( x_e)$, and split that link into two links $e_1,e_2$ with $l_e = l_{e_1}+l_{e_2}$, $x_e=x_{e_1}=x_{e_2}$, and $u_e=u_{e_1}=u_{e_2}$. Then, 
\[l_e u_e x_e - l_e F(x_e) = l_{e_1} u_{e_1} x_{e_1} - l_{e_1} F(x_{e_1} ) +l_{e_2} u_{e_2} x_{e_2} - l_{e_2} F(x_{e_2} ). \]
In other words, the utility contribution of the whole link is equal to the sum of the contributions of the two parts. This is a desirable model feature, as it makes the predictions of the model invariant with respect to the introduction of dummy nodes.

The utility rate vector $u$ is specified as a linear function $u=z\beta$ of link characteristics organized in a matrix $z$ 
and parameters organized in a column vector $\beta$ to be estimated. {\tr Which link characteristics to include in $z$ depends on the specific application. In our empirical application in Section  \ref{sec:application}, we use number of outlinks (divided by link length), travel time per kilometer (i.e., pace), and link type dummies.}

Even though the utility function \eqref{eq:pertutil} does not comprise an explicit representation of heterogeneity, it may still be interpreted as representing an underlying heterogeneous population of travelers {\tr  \citep{Allen2019}, in the same way as the choice probabilities of some additive random utility discrete choice model represent an underlying population of agents with heterogeneous random utility residual terms.   }

\subsection{Solving the traveler's problem}

{\tr We solve the traveler's problem \eqref{eq:pertutil} by setting up the corresponding Lagrangian

\[
\Lambda \left( x,\lambda \right) =l^{\intercal }\left( u\circ x\right)
-l^{\intercal }F\left( x\right) 
+\lambda ^{\intercal }\left( Ax-b\right),
\]%
}
where $\lambda \in \R^{|\mathcal{V}|}$ is a vector of Lagrange multipliers, one for each node in the network, corresponding to the flow conservation constraints. For each link $e$, $x_{e}$ is either zero or the partial derivative of the
Lagrangian with respect to $x_{e}$ is zero. Altogether, we have the first-order conditions
\begin{eqnarray*}
0 &=&\hat x\circ \left( l\circ \left( u-F'(\hat x) \right) +A^{\intercal
}\hat \lambda \right)  \\
A\hat x &=&b,
\end{eqnarray*}
where hats denote the optimal values of $x$ and $\lambda$.

{\tr We can interpret the first-order conditions, noting that for each active link, the marginal utility rate equals the marginal perturbation, that is,  $u_e=F'(\hat x_e)$}, except for the effect induced by the flow conservation constraint. 

Let $B$ be a matrix that is an $|\mathcal{E}|$-dimensional identity matrix, except (at least) all rows corresponding to edges with zero flows are omitted. Pre-multiplying the first-order conditions by $B$, we can disregard links with zero flows (and possibly more), obtaining 

{\tr
\begin{equation}\label{eq:FOC_B}
B\left( l\circ u\right) =B\left( l\circ F'(\hat x) \right)
-BA^{\intercal }\hat \lambda.
\end{equation} }

{\tr We allow that $B$ omits also rows corresponding to some positive flows. We mention this as it may be useful in future applications where there is concern over sampling noise owing to small vehicle counts on some links. We have not made use of this possibility in the results presented in this paper.}  

Equation \eqref{eq:FOC_B} can be used to formulate a regression model. For each origin--destination (OD) pair, the corresponding traveler's problem leads to vectors $( \hat x,\hat \lambda)$,  where the elements of {\tr $B (l \circ F'(\hat x))$ }can be used as  dependent variables.

Noting that 
\[B(l\circ u) = B(l\circ z\beta) = B(l\circ z) \beta, \]
the matrix $B(l\circ z)$ can act as independent variables with corresponding parameter vector $\beta$. 
The Lagrange multipliers $\hat \lambda$ can be treated as fixed effects, that can be corrected for in the regression. That is, however, computationally challenging, as the size of $\hat \lambda$ is equal to the number of nodes in the network and the vector is specific to each OD combination. The number of fixed effects may therefore become very large as the size of the network and the number of OD combinations increases. We will therefore seek an alternative regression equation that eliminates $\hat \lambda$ {\tr from \eqref{eq:FOC_B}}.

\subsection{Eliminating the Lagrange multipliers}

{\tr For eliminating $\hat \lambda$ from \eqref{eq:FOC_B}, we have to take into account that the matrix $BA^{\intercal }$ is not invertible. 
Let }$C=\left( BA^{\intercal }\right) ^{\ast }$ be a Moore--Penrose inverse of $BA^{\intercal }$. We may then utilize that $BA^{\intercal }CBA^{\intercal }=BA^{\intercal }$. Multiplying the traveler's reduced first-order condition \eqref{eq:FOC_B}  by $BA^{\intercal }C$, we find that
{\tr
\begin{eqnarray*}
BA^{\intercal }CB\left( l\circ u\right)  &=&BA^{\intercal }CB\left( l\circ F'(\hat x) \right) -BA^{\intercal }CBA^{\intercal }\lambda  \\
&=&BA^{\intercal }CB\left( l\circ F'(\hat x) \right) -BA^{\intercal
}\lambda,
\end{eqnarray*}%
which leads to
\[
BA^{\intercal }\lambda =BA^{\intercal }CB\left( l\circ \left( F'(\hat x) -u\right) \right).
\]}
Substituting this back into the reduced first-order condition \eqref{eq:FOC_B}, we find that 
{\tr

\begin{equation*}
B\left( l\circ u\right) =B\left( l\circ F'(\hat x) \right)
-BA^{\intercal }CB\left( l\circ \left( F'(\hat x) -u\right) \right) 
\end{equation*}
which leads to 
\begin{equation}\label{eq:FOC_C}\left( I-BA^{\intercal }C\right) B\left( l\circ u\right) =\left( I-BA^{\intercal
}C\right) B\left( l\circ F'(\hat x) \right).
\end{equation} }
We have thus managed to eliminate the Lagrange multipliers from the first-order condition.

\subsection{A regression equation {\tr for model estimation}}\label{sec:reg_eq}
Given a flow vector $\hat x$ corresponding to demand $b$, we now use \eqref{eq:FOC_C} to motivate the construction of a vector {\tr $y=( I-BA^{\intercal }C) B( l\circ F'(\hat x) )$} and a matrix $w =\left( I-BA^{\intercal }C\right) B\left( l\circ z\right)$ with the aim of estimating the parameters of the model by regression. 

Doing this for each demand vector $b$, we construct pairs $(y_{bi},w_{bi})$, where $b\in\mathcal{B}$ indexes the set of demand vectors and $i$ indexes the elements of each vector $y_b$. 

Adding mean zero noise terms $\epsilon_{bi}$, we see that parameters $\beta$ can be recovered from the regression

\begin{equation}\label{eq:OLS}
    y_{bi} = w_{bi} \beta + \epsilon_{bi}.
\end{equation}
We assume the noise terms to be independent of the variables in $w_{bi}$, but allow them to be heteroscedastic, which leads us to use robust standard errors when running regressions. 

The computational complexity of the regression is  independent of the number of route choice observations constituting the flow vector $\hat{x}$ for each $b$, and only increases linearly with the size of the set of demand vectors $\mathcal{B}$, as the computationally most demanding task is obtaining the Moore-Penrose inverse $C$ for each $b$.\footnote{We used the IMqrginv function in Matlab \citep{Ataei2014}, which required less than 0.06 seconds per OD combination.
It is probably possible to utilize sparsity to speed this up.}



\subsection{Routes are substitutes}

The perturbed utility route choice model is formulated in terms of links, which we have seen is useful for estimation and avoids the need to enumerate the set of available routes.  Proposition \ref{prop:substitutes}{, established in this section, }shows that the model can (in principle) equivalently be formulated in terms of routes. {\tr Formulating the model in terms of routes would be impossible in practice, as all routes would have to be enumerated. However, this} perspective allows some understanding to be gained of the model's properties. 

The proposition shows that the equivalent model is a perturbed utility discrete choice model over the universal set of loop-free routes in which the routes are substitutes. An important driver of this result is that loops are ruled out by utility maximization.

{\tr Why is it important to note that routes are substitutes? Most discrete choice models used in transportation are additive random utility models \citep{McFadden1974a, McFadden1981}. A feature of these models is that alternatives are always substitutes in the sense that increasing the utility of one alternative (weakly) decreases the probability that any other alternative is chosen. This property does not, however, hold in general for all discrete choice models and, in particular, it does not hold in general for perturbed utility discrete choice models \citep{Allen2019b}. Depending on the specification of the convex perturbation function, these models can allow alternatives to be complements such that increasing the utility of one alternative increases the probability that another alternative is chosen. This is important in practice, as there are many situations in which alternatives are complements. In a supermarket, we may think of products that are often consumed together (such as salsa and nachos). In route choice, many complementarities occur at the level of links.  It is therefore quite informative and important for the understanding of the present model to establish that it predicts that routes are substitutes.}

\begin{proposition}\label{prop:substitutes}
The perturbed utility route choice model is equivalent to a discrete choice model for the choice among all loop-free routes connecting origin to destination.  The routes are substitutes in the sense that if the utility of just one route is increased while the utility of all other routes is unaffected, then the choice probability weakly decreases for all other routes.
\end{proposition}
\begin{proof}[Proof of Proposition \ref{prop:substitutes}]
Denote the set of all routes connecting origin to destination by $\Sigma $
and a single route as a list of links $\sigma =\left\{ e_{\sigma \left(
1\right) },e_{\sigma \left( 2\right) },\ldots\right\} $. Any flow-conserving
flow vector can be written as a sum of route flows. The utility function %
\eqref{eq:pertutil} can then be written in terms of vectors $p=\left\{ p_{\sigma },\sigma \in \Sigma \right\} \subseteq \mathbb{R}^{|\Sigma|}_+$ as
\begin{eqnarray}\label{eq:route_util}
U\left( p\right)  &=&\tsum\limits_{\sigma \in \Sigma }p_{\sigma }u_{\sigma
}-G\left( p\right) ,\text{ where} \\ 
u_{\sigma } &=&\tsum\limits_{e\in \sigma }u_{e},\quad G\left( p\right)
=\tsum\limits_{e\in \mathcal{E}}l_{e}F\left( \tsum\limits_{\sigma \ni
e}p_{\sigma }\right) ,
\end{eqnarray}%
and note that we do not yet restrict the domain of $G$ to probability vectors. The utility is strictly decreasing in each element of $p$, which means no loss of generality when restricting $\Sigma $ to be finite, consisting only
of routes with no loops. The function $F$ is strictly convex by assumption
and hence $G$ is convex. In fact, it is strictly convex as shown by the following argument. 

Assume $p^{1}\neq
p^{2} \in \mathbb{R}^{|\Sigma|}_+$, then%
\begin{eqnarray*}
G\left( \eta p^{1}+\left( 1-\eta \right) p^{2}\right)  &=&\tsum\limits_{e\in
\mathcal{E}}l_{e}F\left( \eta \tsum\limits_{\sigma \ni e}p_{\sigma
}^{1}+\left( 1-\eta \right) \tsum\limits_{\sigma \ni e}p_{\sigma
}^{2}\right)  \\
&\leq &\eta \tsum\limits_{e\in \mathcal{E}}l_{e}F\left( \tsum\limits_{\sigma
\ni e}p_{\sigma }^{1}\right) +\left( 1-\eta \right) \tsum\limits_{e\in
\mathcal{E}}l_{e}F\left( \tsum\limits_{\sigma \ni e}p_{\sigma }^{2}\right)
\\
&=&\eta G\left( p^{1}\right) +\left( 1-\eta \right) G\left( p^{2}\right) ,
\end{eqnarray*}%
where the inequality is strict for links $e$ with $\tsum\limits_{\sigma \ni
e}p_{\sigma }^{1}\neq \tsum\limits_{\sigma \ni e}p_{\sigma }^{2}$. However, such links exist since $p^{1}\neq p^{2}$ and hence $G$ is strictly convex.

The convex conjugate of $G$ restricted to the set of probability vectors is the indirect perturbed utility
\[
U^{\ast }\left( u\right) =\sup_{p\in \Delta \left( \Sigma \right) }\left\{
U\left( p\right) \right\} =\sup_{p\in \Delta \left( \Sigma \right) }\left\{
\tsum\limits_{\sigma \in \Sigma }p_{\sigma }u_{\sigma }-G\left( p\right)
\right\} .
\]%
If $G$ is supermodular, then its convex conjugate $U^{\ast }$ is supermodular \citep[e.g.,][Lem. 1]{Feng2018}. So consider the mixed partial derivatives of the convex perturbation,%
\[
\frac{\partial ^{2}G\left( p\right) }{\partial p_{\sigma _{1}}\partial
p_{\sigma _{2}}}=\tsum\limits_{e\in \mathcal{E}}l_{e}\frac{\partial
F^{\prime }\left( \tsum\limits_{\sigma \ni e}p_{\sigma }\right) 1_{\left\{
e\in \sigma _{1}\right\} }}{\partial p_{\sigma _{2}}}=\tsum\limits_{e\in
\mathcal{E}}F^{\prime \prime }\left( \tsum\limits_{\sigma \ni e}p_{\sigma
}\right) 1_{\left\{ e\in \sigma _{1}\cap \sigma _{2}\right\} },
\]%
to conclude they are either zero for routes that do not overlap or
negative. Hence, the convex perturbation $G$ is submodular and $U^*$ is supermodular. 
Moreover, $U^{\ast }$ satisfies the definition of a choice welfare function
in \citet{Feng2018}. Hence, by \citet[][Thm. 1]{Feng2018}, if $U^{\ast }$ is
differentiable, then all routes are substitutes. However, this holds by \citet[][Thm. 26.3]{Rockafellar1970}, as $G$ is essentially strictly convex.\hfill
\end{proof}

The proof of Proposition \ref{prop:substitutes} shows that the perturbed utility model has an equivalent formulation as a perturbed utility model  \eqref{eq:route_util} for the choice among all loop-free routes connecting origin to destination. The convex perturbation $G$ is a length-weighted sum across links of terms that are the convex function $F$ applied to each link flow. The flow on link $e$ is the sum of route choice probabilities for routes using that link,  $\tsum\limits_{\sigma \ni
e}p_{\sigma }$. The presence of such sums in the utility function generates a tendency for substitutability. In fact, if there was only one such  term $F\left(\tsum\limits_{\sigma \ni
e}p_{\sigma }\right)$, then routes using link $e$ would be perfect substitutes. {\tr We explore the model's predicted substitution patterns in Section \ref{sec:experiments} below.}

{ \tr 

\subsection{Functional form for the perturbation function}
The shape of the perturbation function matters for the pattern of substitution across links with different flows. To see this, consider how the optimal flow vector is affected by a change in the utility of some active link. The first-order condition \eqref{eq:FOC_B} involves the derivatives $F'(\hat x_e)$, which are increasing in link flows. Hence, smaller changes in link flows are required to adjust the first-order condition on links where the flow is high than on links where the flow is low. The size of the difference is determined by the degree of convexity of the perturbation function. 

To apply the perturbed utility model, we need to impose a specific form on the perturbation function while respecting Assumption \ref{ass:F}. Taking inspiration from the well-known connection between the multinomial logit (MNL) model and the Shannon entropy \citep[e.g.,][]{Hofbauer2002}, we use the form of the entropy function, modified to satisfy the requirements of Assumption \ref{ass:F}, and define
\begin{equation}\label{eq:modentropy}
F(x) = \left(
\left( 1+x\right) \cdot \ln \left( 1+x\right) -x\right),  
\end{equation}
where  $F'(x) = \ln(1+x)$ and $F''(x)=(1+x)^{-1}$, such that $F'(0)=0$ and $F''(x)>0$ for $x>0$ as required. In our empirical application in Section \ref{sec:application}, we also tested a quadratic perturbation $F(x)=x^2$, but found that the entropy-like function \eqref{eq:modentropy} led to a better fit. 
}

\section{Experimenting with the model}\label{sec:experiments}

In this section, we investigate the behavior of the model by solving the traveler's problem, first on a small toy network and then on a large network for the  Copenhagen metropolitan area.

\subsection{A toy example} \label{sec:ToyExample}

\begin{figure}[h]
    \centering
    \includegraphics[width=0.5\textwidth]{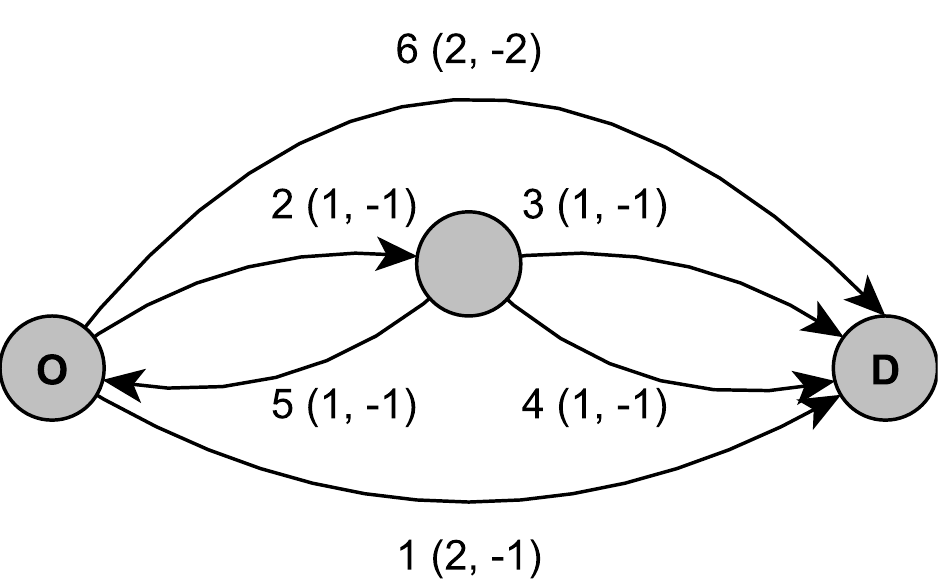}
    \caption{Toy network. Numbers in parentheses indicate length and unit cost, i.e., ($l_e$, $u_e$).}
    \label{fig:ToyNetwork}
\end{figure}

Consider a toy network with six links and four loop-free routes from origin to destination; see Figure \ref{fig:ToyNetwork}. One route uses link 1 with length $l_1=2$, which goes straight from the origin to the destination. Two routes share link 2 with length $l_2=1$ and then split into links 3 and 4, each also with length $l_3=l_4=1$. The fourth route uses link 6 with length $l_6=2$ that also directly connects the origin to the destination.  
To test the ability of the perturbed utility model to handle {\tr the possibility of loops}, we add link 5 that is equal to link 2 but goes in the opposite direction, such that it is consistent with flow conservation to have flow looping on links 2 and 5.  The unit link utility rate is $u_e=-1$ on all links except link 6, on which $u_6=-2$. 
This causes three of the four direct routes to have the same total disutility, and the alternative using link 6 to be twice as costly.

The base case utility maximizing flows for the PURC model are shown in Table \ref{tab:toy_example}. We have also included flows for a corresponding MNL model \citep{McFadden1974a}  and a PSL model  \citep{Ben-Akiva1999PSL}.\footnote{ {\tr Parameters for the MNL and PSL models were calibrated against PURC flow. As link 6 was unused in the base network for the PURC, the calibration was performed on a slightly modified network with $u_6=-1.25$, causing link 6 to be used for all models. 
The direct utilities of alternative $n$ for the MNL and PSL models were defined as $V_n = \beta_u U_n$ and $V_n = \beta_u U_n + \beta_\text{PS} \ln S_n$, respectively, where $U_n = \sum\limits_{e \in n} l_e u_e$, and $S_n$ was defined as in \cite{Ben-Akiva1999PSL} based on link lengths. In both cases, the probability of choosing an alternative $n$ within the choice set $\mathcal{C}$ was given by $P_n = \frac{e^{V_n}}{\sum\limits_{k \in \mathcal{C}} e^{V_k}}$.
Setting $\beta_u=1$ for the PURC model and d}isallowing loops for the MNL and PSL models, the parameters found to reduce the sum of square errors across the link flows for the MNL and PSL models were
$\beta_u = 2.0$ and $\beta_\text{PS} = 1.1$, respectively.
}
We note that the PURC model assigns zero flow on link 5; hence, utility maximization causes the loop to not occur even though it is (physically) possible in the model. In contrast, the MNL and PSL models rely on generated choice sets for which we omitted paths with loops. {\tr We note also that the PURC model assigns zero flow on the costly link 6, which demonstrates that the model is in fact able to assign zero flow.} In contrast, the MNL and PSL models both assign a positive flow because the route using link 6 is included in the choice set.

\begin{table}[ht]
    \centering
    \begin{tabular}{|l|l|ll|ll|ll|ll|}  \hline
     & Link  &  $l_e$ & $u_e$    & \multicolumn{2}{l|}{PURC flow}  & \multicolumn{2}{l|}{MNL flow}  & \multicolumn{2}{l|}{PSL flow}    \\ \hline
    Base &   $1$ & 2 & -1 &  0.424 & & 0.331 & & 0.404 &\\
    &   $2$     &  1 & -1 &  0.576 & & 0.663 & & 0.589 &\\
    &   $3$     &  1 & -1 &  0.288 & & 0.331 & & 0.294 &\\
    &   $4$     &  1 & -1 &  0.288 & & 0.331 & & 0.294 &\\ 
    &   $5$     &  1 & -1 &  0 & & 0\footnote & & 0\footnotemark[\value{footnote}]  &\\ 
    &   $6$     &  2 & -2 &  0 & & 0.006 & & 0.007 &\\ 
    \hline
   Increase unit cost  &    $1$     &  2 & -1 &  0.445 & \textit{(1.047)} & 0.352 & \textit{(1.064)} & 0.427 & \textit{(1.056)}\\
    on link 4&   $2$     &  1 & -1 & 0.555 & \textit{(0.965)} & 0.641 & \textit{(0.967)} & 0.566 & \textit{(0.961)}\\
   \textit{(rel. to base)}&   $3$     &  1 & -1 & 0.342 & \textit{(1.187)} & 0.352 & \textit{(1.064)} &  0.311 & \textit{(1.056)}\\
    &   $4$     &  1 & -1.1 & 0.214 & \textit{(0.743) }& 0.289 & \textit{(0.871)} & 0.255 & \textit{(0.865)}\\ 
    &   $5$     &  1 & -1 & 0 & - & 0\footnotemark[\value{footnote}]& - & 0\footnotemark[\value{footnote}] & -\\ 
    &   $6$     &  2 & -1 & 0& - & 0.006 & \textit{(1.064)} & 0.008 & \textit{(1.056)}\\ 
    \hline
    Move node at the &   $1$     &  2 & -1 & 0.381 & \textit{(0.897)} & 0.331 & \textit{(1)} & 0.364 & \textit{(0.902)}\\
    end of link 2&   $2$     &  0.5 & -1 & 0.619 & \textit{(1.076)} & 0.663 & \textit{(1)} & 0.629 & \textit{(1.069)}\\
     \textit{(rel. to base)}&   $3$      &  1.5 & -1 & 0.310 & \textit{(1.076)} & 0.331 & \textit{(1)} & 0.315 & \textit{(1.069)}\\
    &   $4$     &  1.5 & -1 &  0.310 & \textit{(1.076)}  & 0.331 & \textit{(1)} & 0.315 & \textit{(1.069)}\\ 
    &   $5$     &  0.5 & -1 & 0 & - & 0\footnotemark[\value{footnote}] & - & 0\footnotemark[\value{footnote}] & -\\ 
    &   $6$     &  2 & -2 & 0 & - & 0.006 & \textit{(1)} & 0.007 & \textit{(0.902)}\\ \hline
    \end{tabular}
    \caption{Utility maximizing flows in the toy network. {\tr Numbers in parentheses indicate flows relative to base.}}
    \label{tab:toy_example}
\end{table}
\footnotetext{Zero flow is purely a consequence of the choice set generation method.}

We consider two experiments. In the first, we increase the unit cost on link 4 by 0.1. In all three models, this leads to a decrease in the flow on the route using link 4. In the PURC model, the flow on the route using link 3 increases more in relative terms than the flow on link 1. Thus, the IIA property does not hold for the PURC model: the routes using links 3 and 4 are closer substitutes with each other than with the route using link 1. This is a desirable property, which occurs because the two close routes share link 2.
In contrast, in the MNL and PSL models, the flow on the routes using links 1, 3, and 6 increase by the same proportion, in accordance with the IIA property.

In the second experiment, we reduce the lengths of links 2 and 5 by 0.5 and increase the lengths of links 3 and 4 by the same amount. This makes the routes using link 2 more dissimilar and hence attracts more flow to them in the PURC model. The limiting cases of this kind of change are as desired: the utility maximizing flows split evenly on the three routes if link 2 is reduced to zero length, and the utility maximizing flows split fifty-fifty if the length of link 2 is increased to 2. In other words, as the two routes overlap less, they function as more independent routes. In contrast, the change in link lengths leads to no change in flows in the MNL model. The PSL model reacts similarly to the PURC model in this case.

{\tr We summarize the points in the  comparison of the PURC model to the MNL and PSL models in the following table. The comparison speaks clearly in favor of the PURC model.

\begin{table}[ht]
    \centering
    \begin{tabular}{|l|c|c|c|}
    \hline     & PURC & MNL & PSL \\\hline
    
      Choice set generation  necessary& \smiley{} & \frownie{} & \frownie{} \\
      Loops & \smiley{} & \frownie{} & \frownie{} \\
      Assign zero flow & \smiley{} & \frownie{} & \frownie{} \\
      IIA & \smiley{} & \frownie{} & \frownie{} \\
      Overlap & \smiley{} & \frownie{} & \smiley{}\\
       
    \hline    
    
    \end{tabular}
    \caption{\tr Comparison of models for the toy example.}
    \label{tab:comparison}
\end{table}
}

\subsection{Illustration using the Copenhagen metropolitan area road network}

We proceed to illustrate how the perturbed utility model behaves on a large network.  We use a network, shown in Figure \ref{fig:NetworkRoadTypes}, for the Copenhagen metropolitan area, comprising 30,773 links and 12,876 nodes \citep{Kjems2019}.

\begin{figure}[h]
    \centering
    \includegraphics[width= 8cm]{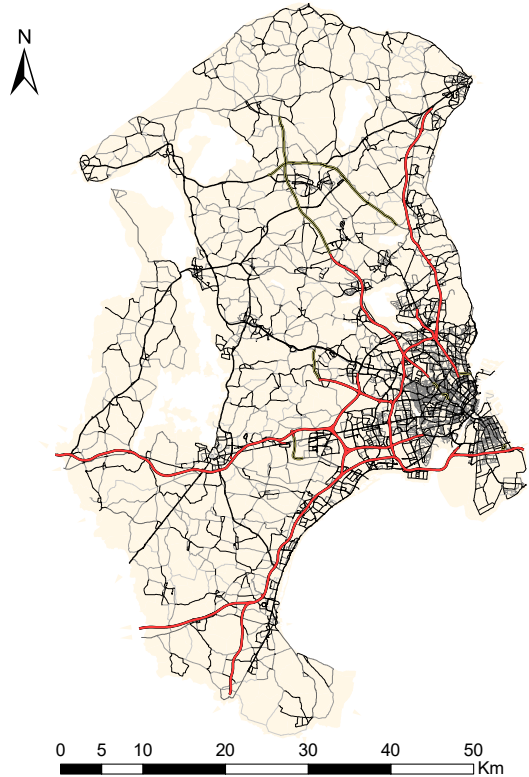}
    \caption{Network of the Copenhagen metropolitan area.}
    \label{fig:NetworkRoadTypes}
\end{figure}

Figure \ref{fig:EffectOfBetaOnFlows} shows the model's prediction for a trip from the airport (located South-East on the figure) to the Technical University of Denmark (North). For the example, we specified the link utility as a parameter $\beta$ multiplied by {\tr pace (defined as travel time in minutes per kilometer)} and used different values of $\beta$.  We solved the traveler's problem using MatLab \citep{MATLAB2010} with the  \texttt{fmincon} command using an interior point algorithm \citep{Byrd1999,Waltz2006} approach for constrained optimization.

Note that most links in the network are inactive, having zero predicted flow. {\tr This is in line with \citet{Rasmussen2017}, who also demonstrated many unused links, but using a model requiring path enumeration.} As in Proposition \ref{prop:substitutes}, the active links and predicted flows can be translated into a set of active routes with corresponding route choice probabilities. 



The model predicts two 
clear main alternatives.
One is the Western motorway bypass across Kalvebod bridge (Kalvebodbroen), and the other is a cluster of alternatives going through the Copenhagen city center. 

{\tr The optimal flow balances the incentive to reduce travel time, determined by the size of $\beta$, against the incentive to distribute flow across more links in the network, determined by the perturbation $F$. As Figure \ref{fig:EffectOfBetaOnFlows} illustrates, the incentive to distribute flow matters relatively more; hence, the number of active routes is high when $\beta$ is numerically small. As  $\beta$ increases numerically, the incentive to reduce travel time receives more weight and hence link usage approaches the shortest (fastest) path. In particular, the use of the cluster of alternatives through the city center is reduced when $\beta$ becomes numerically larger.}

\begin{figure}
    \centering
    \begin{subfigure}[b]{0.49 \textwidth}
     \centering
         \includegraphics[width=\textwidth]{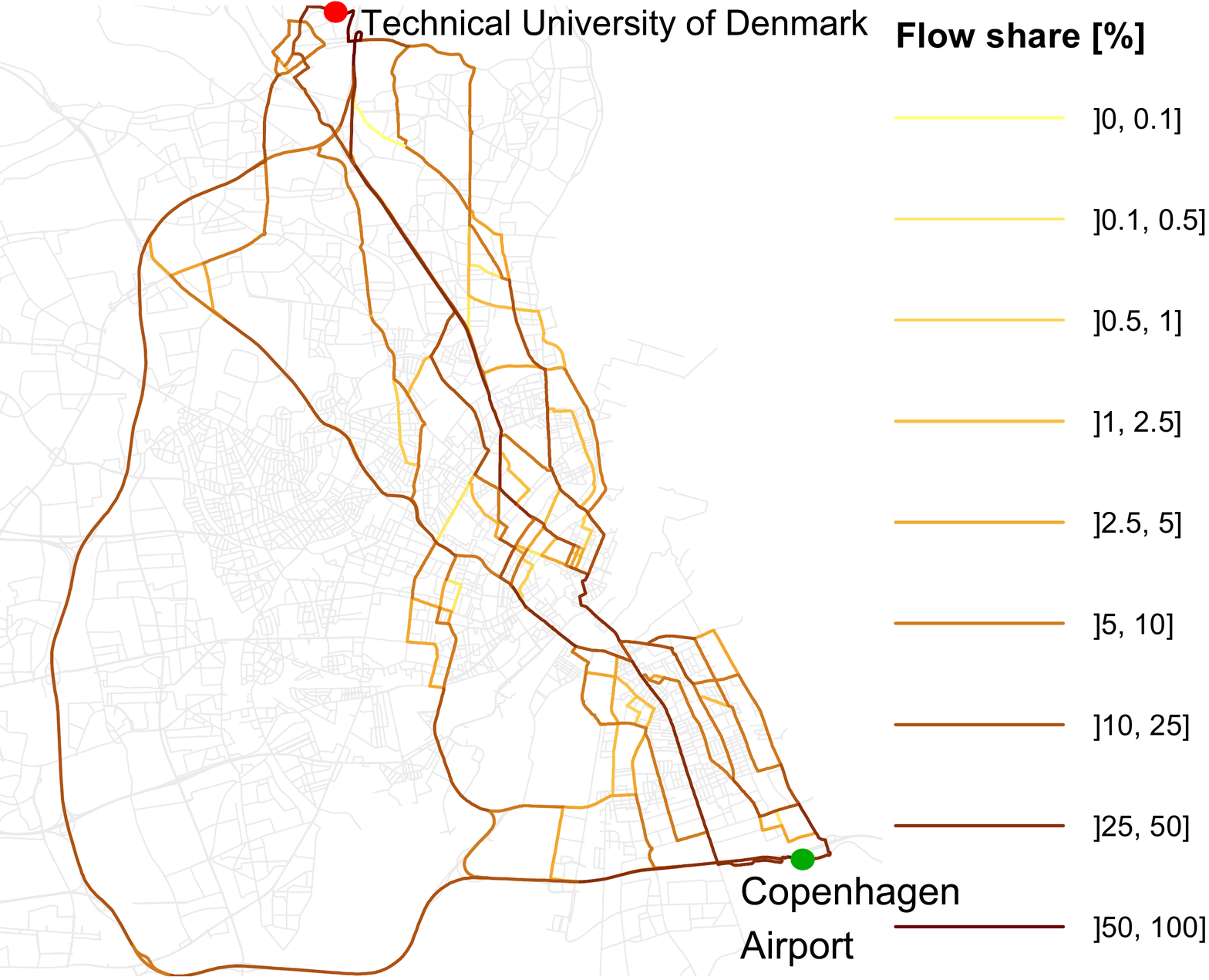}
         \caption{$\beta=-0.5$}
         \label{fig:EffectOfBetaOnFlows_0d5}
    \end{subfigure}
  \hfill
      \begin{subfigure}[b]{0.49 \textwidth}
     \centering
         \includegraphics[width=\textwidth]{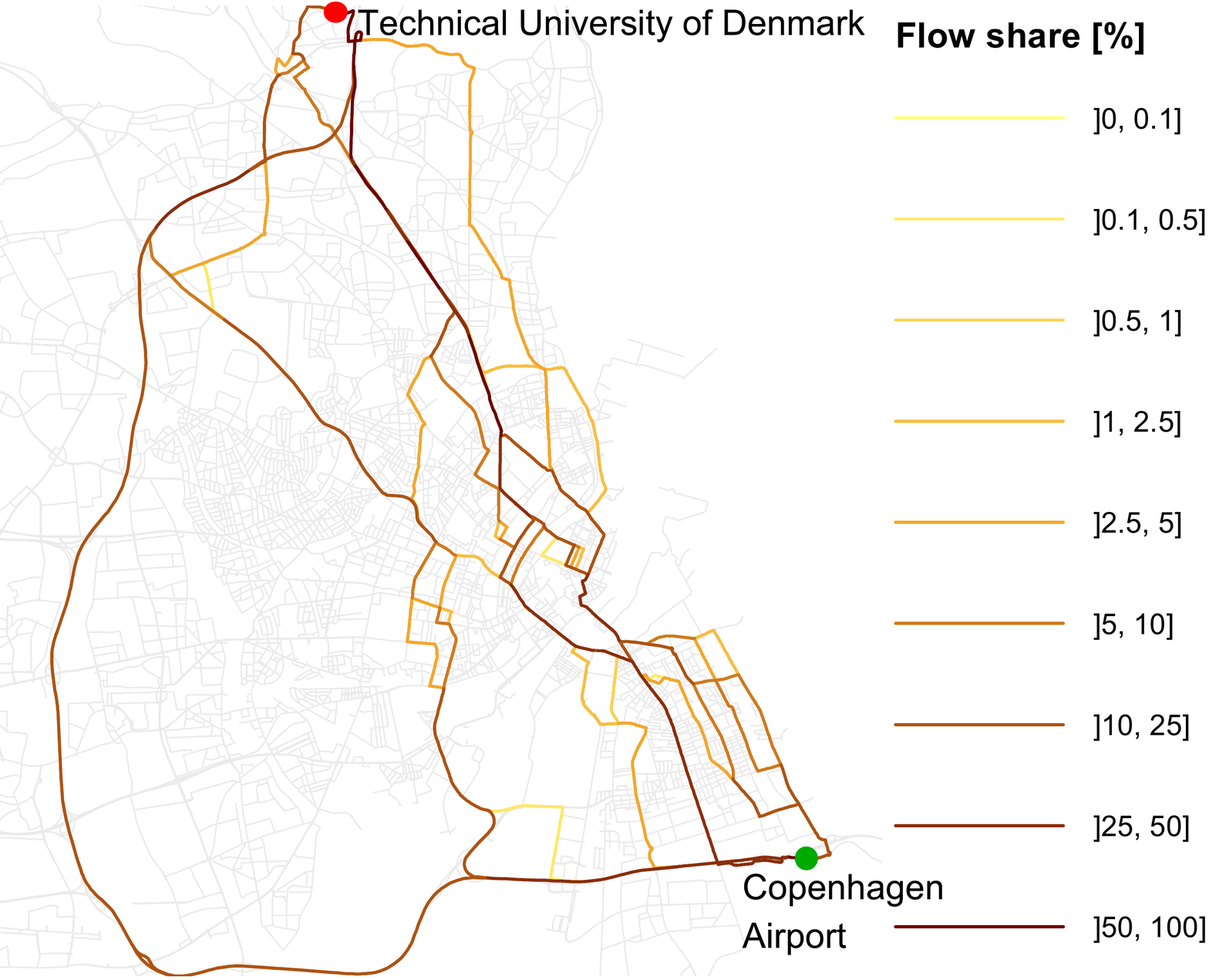}
         \caption{$\beta=-1$}
         \label{fig:EffectOfBetaOnFlows_1}
    \end{subfigure}
\\
    \centering
    \begin{subfigure}[b]{0.49 \textwidth}
     \centering
         \includegraphics[width=\textwidth]{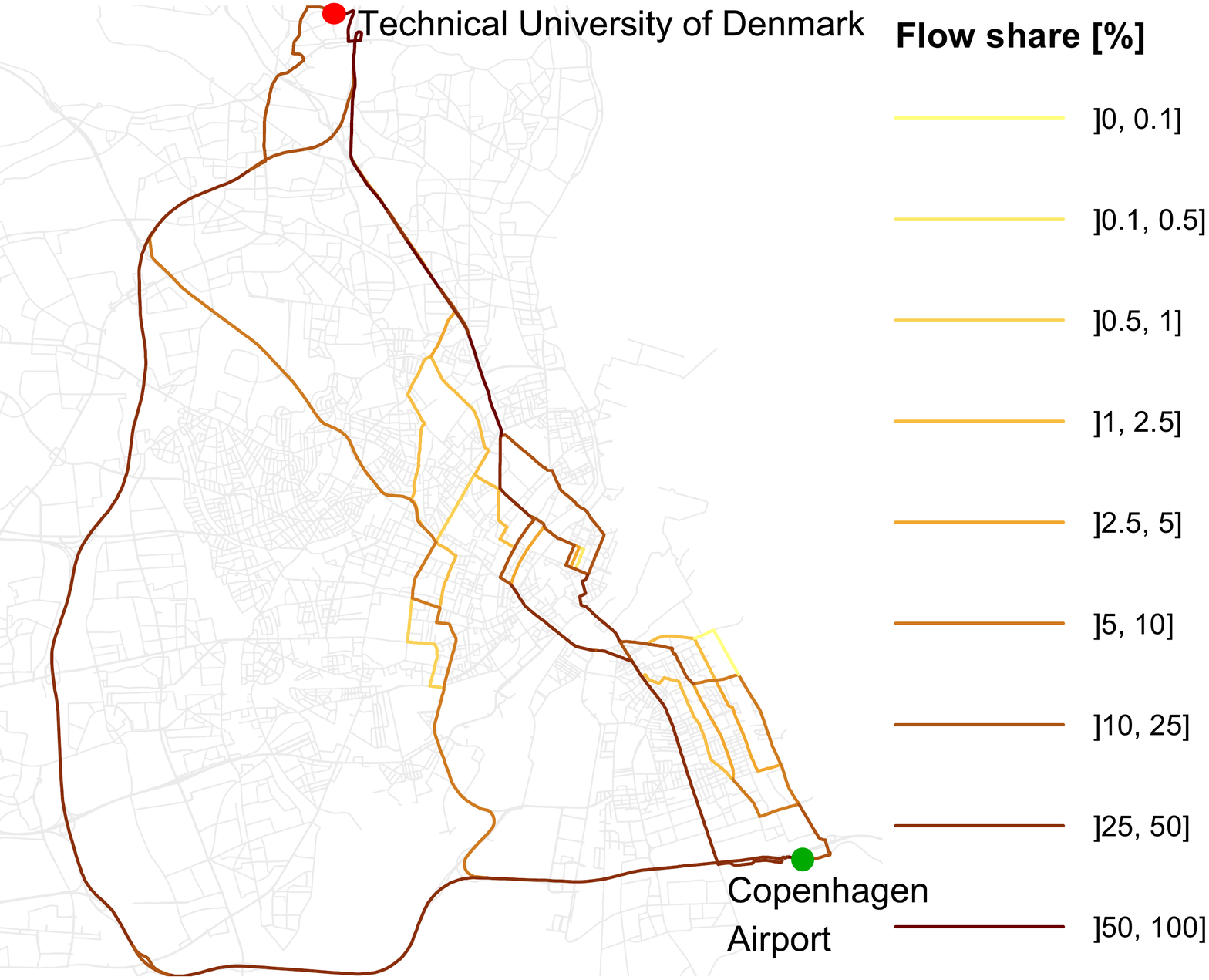}
         \caption{$\beta=-1.5$}
         \label{fig:EffectOfBetaOnFlows_1d5}
    \end{subfigure}
  \hfill
      \begin{subfigure}[b]{0.49 \textwidth}
     \centering
         \includegraphics[width=\textwidth]{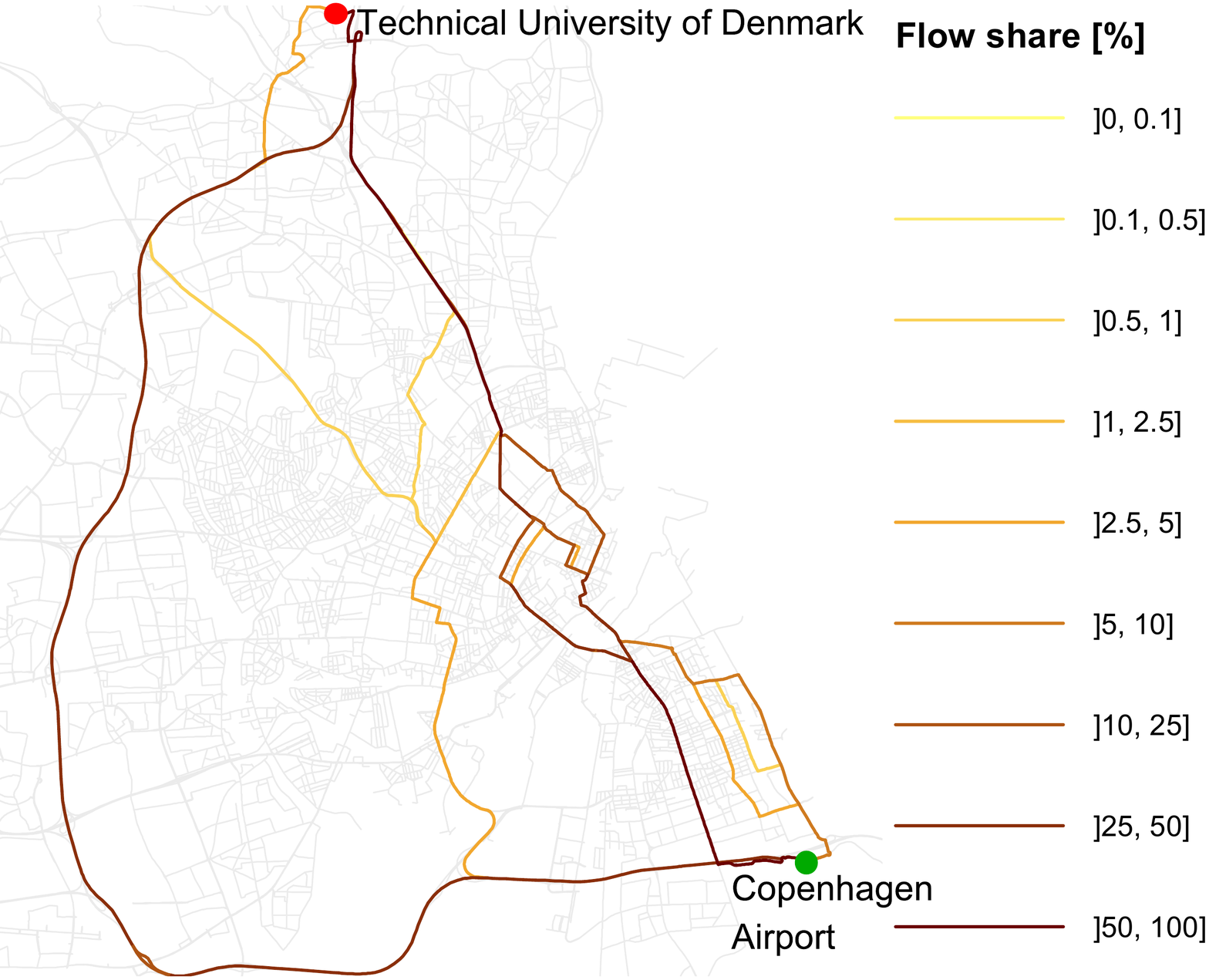}
         \caption{$\beta=-2$}
         \label{fig:EffectOfBetaOnFlows_2}
    \end{subfigure}
    \caption{{\tr Link flows using the PURC model with unit demand} from Copenhagen Airport to Technical University of Denmark for different values of $\beta$.}
    \label{fig:EffectOfBetaOnFlows}
\end{figure}

Figure \ref{fig:CopenhagenSubstitution_Airport_DTU} shows the change in predicted flows that follows an increase {\tr of two minutes in the travel time on Kalvebodbroen in the southern part of the map}. This change increases the cost on all routes that use Kalvebodbroen and consequently flow shifts toward alternative routes. Routes that share links with routes that use Kalvebodbroen are {\tr seen to be }closer substitutes, which confirms the theoretical expectation.

\begin{figure}
    \centering
    \includegraphics[width=0.8\textwidth]{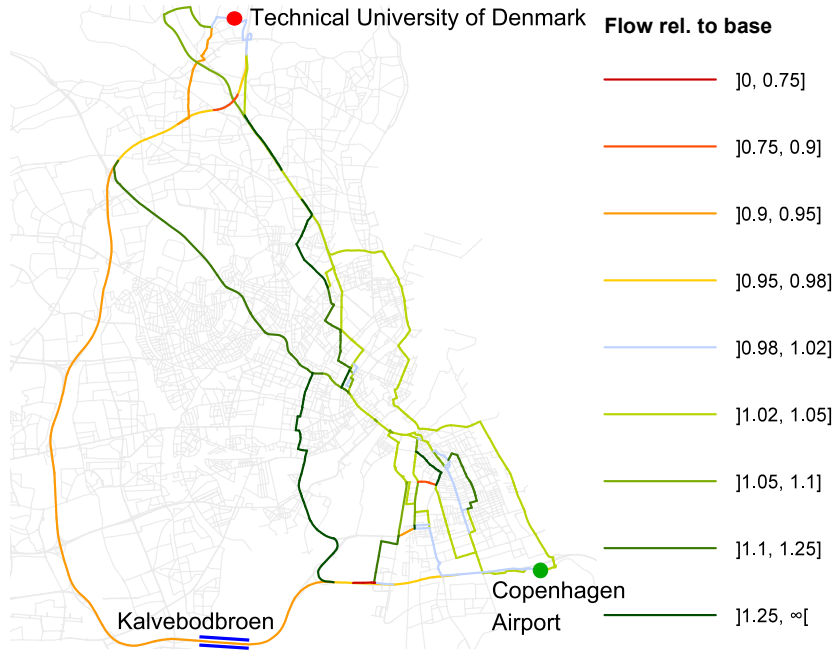}
    \caption{Substitution patterns for trips between Copenhagen Airport and Technical University of Denmark using {\tr the PURC model and} parameters estimated in model C {\tr when} increasing travel time on Kalvebodbroen by two minutes{, \tr relative to the base scenario with no alterations of travel times}.}
    \label{fig:CopenhagenSubstitution_Airport_DTU}
\end{figure}

\section{Testing the estimator on simulated data for Copenhagen}\label{sec:estsimdata}
{\tr Before we proceed to estimate the PURC model on real data, it is useful to explore what precision and bias we might expect. In this section, we therefore test the ability of the estimator to recover a known true parameter from data simulated from the model using the Copenhagen metropolitan area network. We can then assess the quality of the parameter estimates, depending on the size of the dataset. } 

We pre-select 22 nodes {\tr spread across the model area} that are used as origins/destinations. Sampling from these, we generate a number of datasets using a different number of OD combinations (1, 5, 20, and 100) and different values of a parameter for {\tr  pace (minutes per kilometer)
} ($\beta\in\{-3,-2.5,\ldots,-0.5\}$). For each OD and $\beta$, we solve the traveler's problem and sampled a number of {\tr trips }consistent with the predicted link flows.
{\tr Each sampled trip is a random walk between the origin and destination nodes, with link selection probabilities at each node that are proportional to the predicted link flows.} This mimics actual data collection but in a case where the true model is known. We have created datasets sampling, respectively, 25, 100, 250, and 1000 {\tr trips }for each OD.  

{\tr 
Figure \ref{fig:cdfsim} shows cumulative distributions of some summary statistics of the flows underlying the simulation --- that is, the predicted flows found when using the true parameters.  Panel (a) shows how the number of active links increases as the {\tr  parameter $\beta$ for pace} approaches zero and the influence of travel time diminishes. Panel (b) shows that, as expected, a numerically smaller {\tr parameter for pace } increases the average predicted travel time. 

\begin{figure}[h]
    \centering
    \begin{subfigure}{0.49 \textwidth}
        \includegraphics[width = \textwidth]{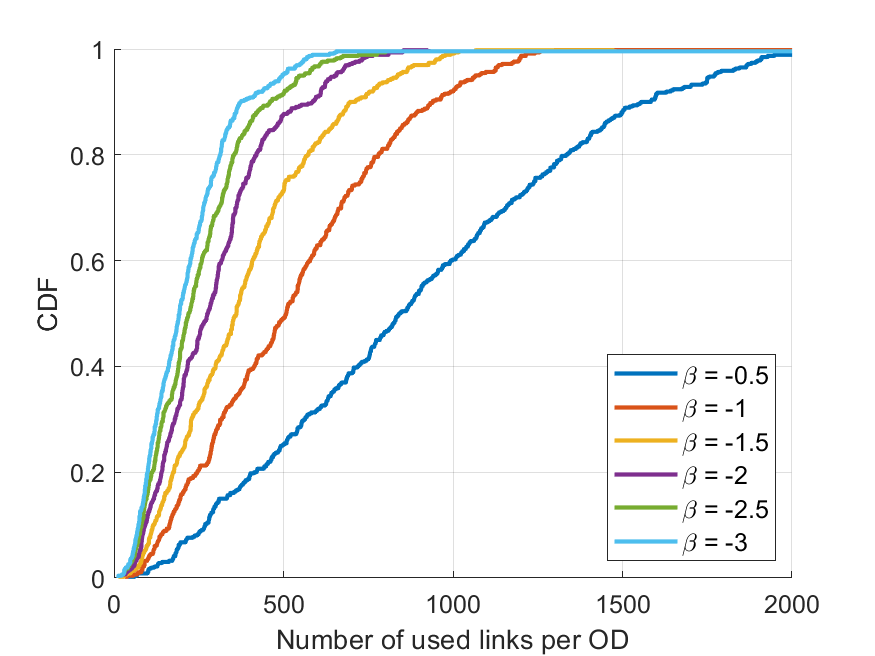}
        \caption{Number of active links per OD}
        \label{fig:UsedLinksPerOD}
    \end{subfigure}
    \begin{subfigure}{0.49 \textwidth}
        \includegraphics[width = \textwidth]{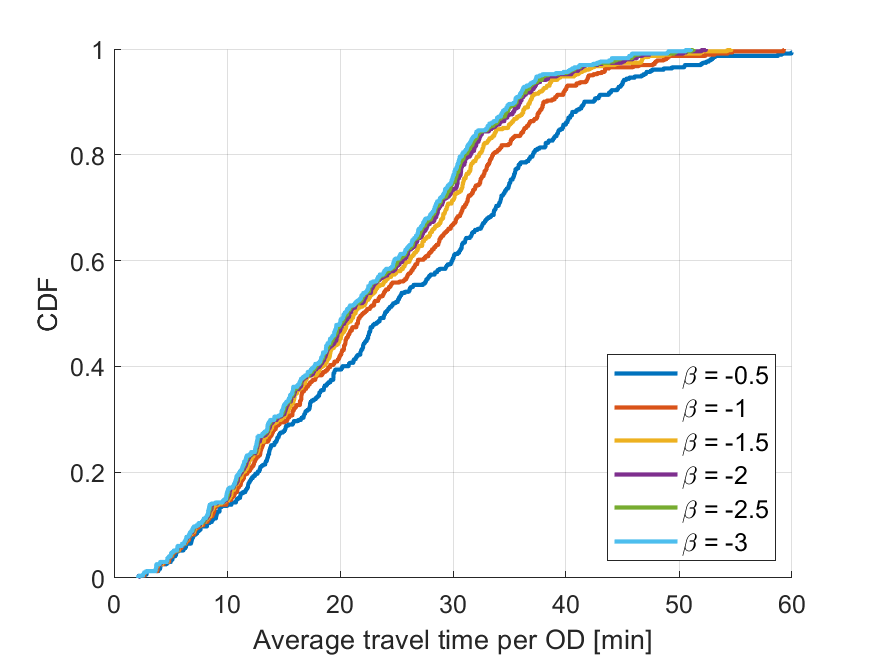}
        \caption{Travel time per OD}
        \label{fig:TripsPerOD}
    \end{subfigure}
    \caption{Cumulative distribution for simulated route choices in Figure {\tr\ref{fig:estimateBetaSim}}. } 
    \label{fig:cdfsim}
\end{figure}

}

Altogether, we have created $4\times 4 \times 6 = 96$ datasets. For each dataset, {\tr treating it as if it were a real dataset of observed route choices}, we have computed a flow vector $x$ for each origin-destination.  
We have then transformed the data as described in Section \ref{sec:reg_eq} and estimated $\beta$ using ordinary least squares (OLS).

\begin{figure}[h]
                \centering
    \begin{subfigure}[b]{0.245 \textwidth}
     \centering
         \includegraphics[width=\textwidth]{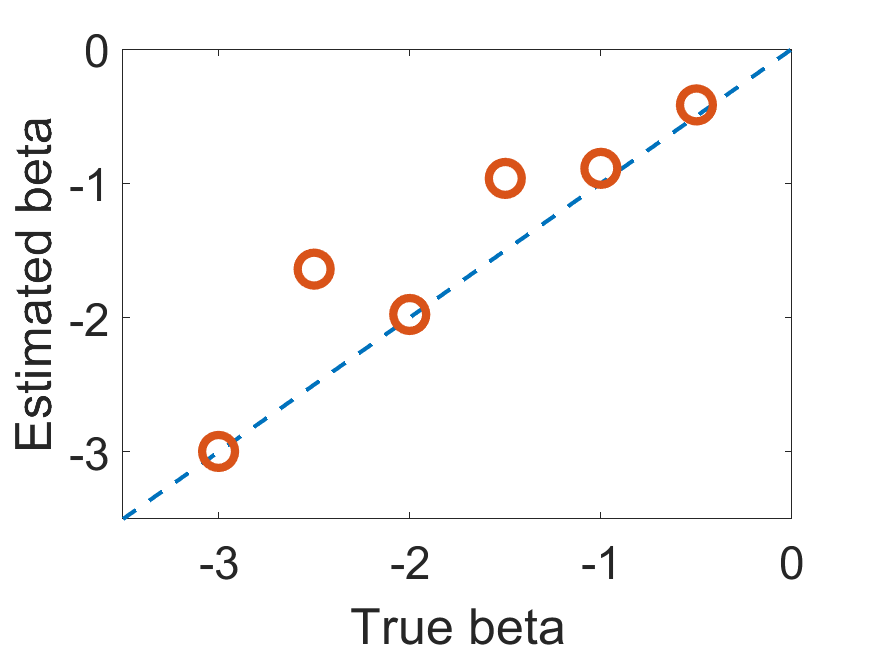}
         \caption{\scriptsize 1 OD, 25 trips per OD}
         \label{fig:Estimations_1ODs_25TripsPerOD}
    \end{subfigure}
  \hfill
      \begin{subfigure}[b]{0.245 \textwidth}
     \centering
         \includegraphics[width=\textwidth]{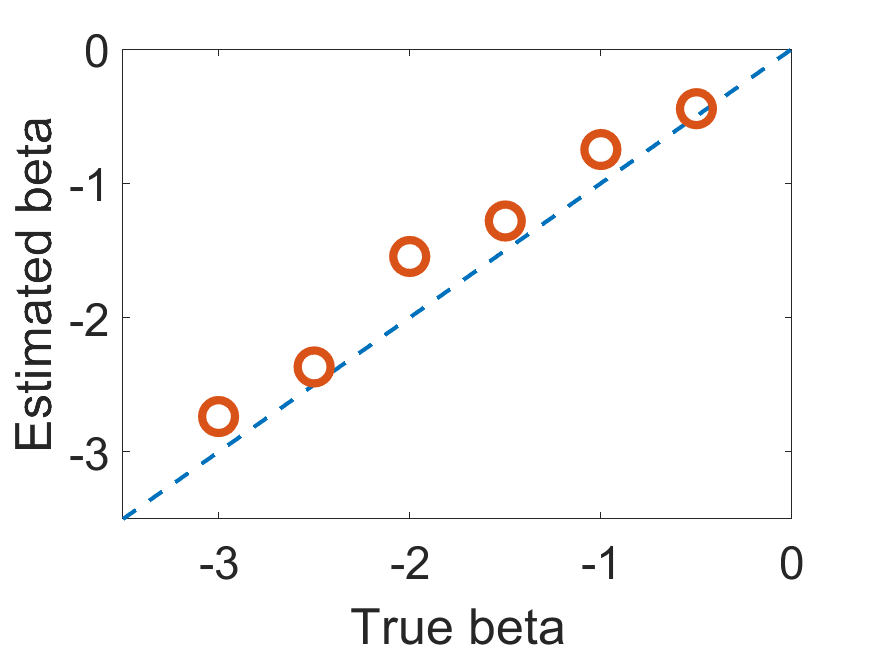}
         \caption{\scriptsize 5 ODs, 25 trips per OD}
         \label{fig:Estimations_5ODs_25TripsPerOD}
    \end{subfigure}
    \begin{subfigure}[b]{0.245 \textwidth}
     \centering
         \includegraphics[width=\textwidth]{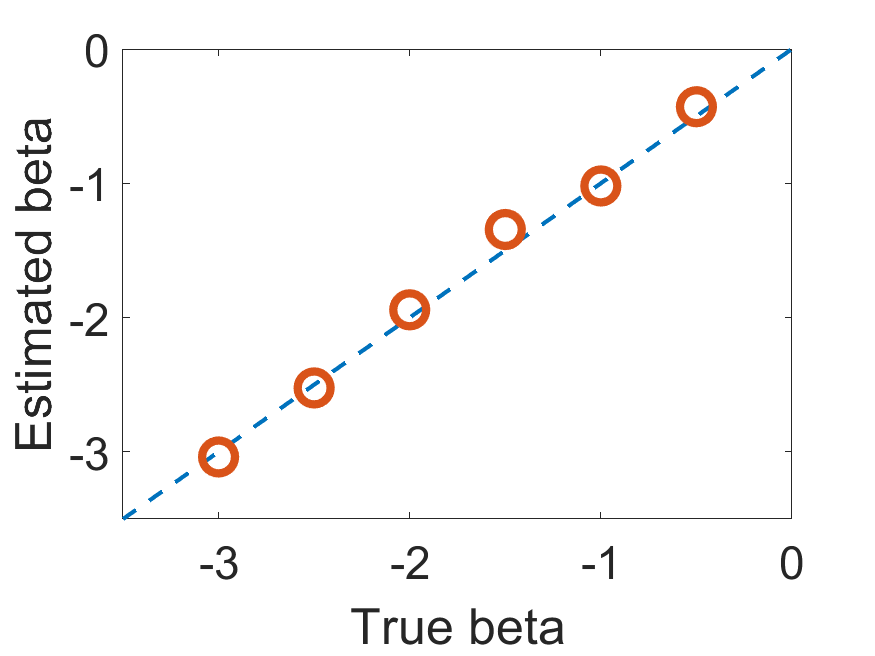}
         \caption{\scriptsize 20 ODs, 25 trips per OD}
         \label{fig:Estimations_2ODs_25TripsPerOD}
    \end{subfigure}
      \begin{subfigure}[b]{0.245 \textwidth}
     \centering
         \includegraphics[width=\textwidth]{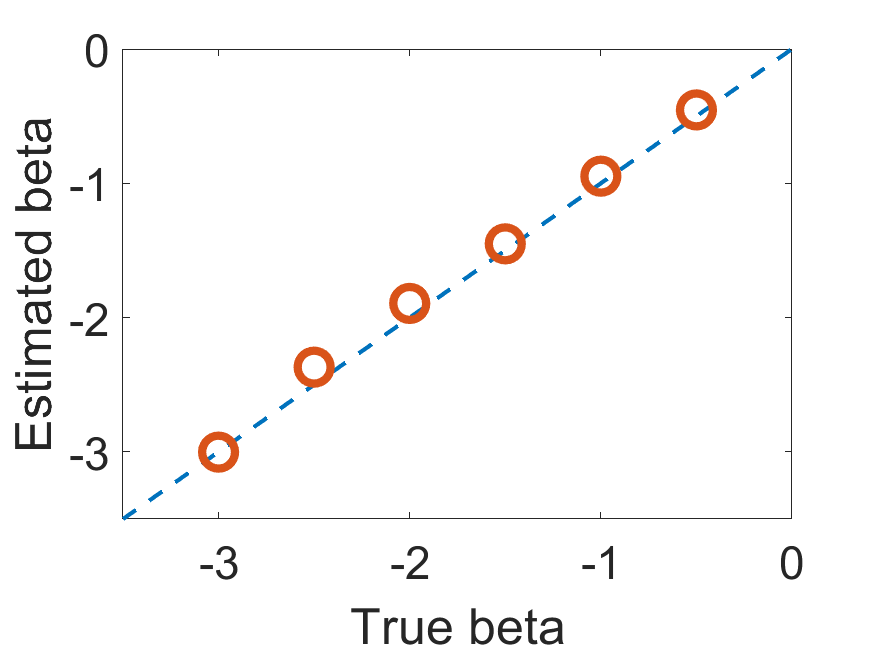}
         \caption{\scriptsize 100 ODs, 25 trips per OD}
         \label{fig:Estimations_100ODs_25TripsPerOD}
    \end{subfigure}
    \\
            \centering
    \begin{subfigure}[b]{0.245 \textwidth}
     \centering
         \includegraphics[width=\textwidth]{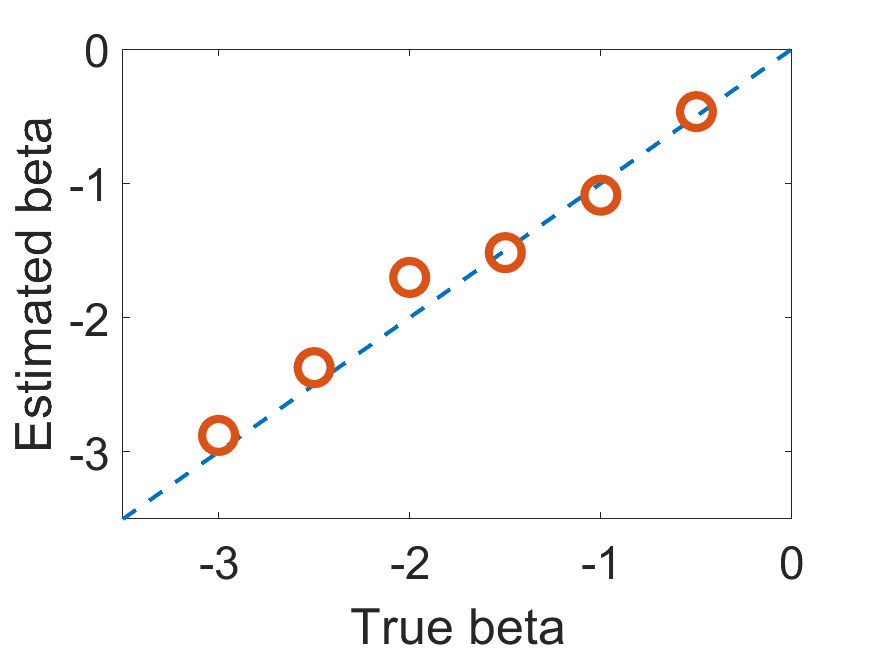}
         \caption{\scriptsize 1 OD, 100 trips per OD}
         \label{fig:Estimations_1ODs_100TripsPerOD}
    \end{subfigure}
  \hfill
      \begin{subfigure}[b]{0.245 \textwidth}
     \centering
         \includegraphics[width=\textwidth]{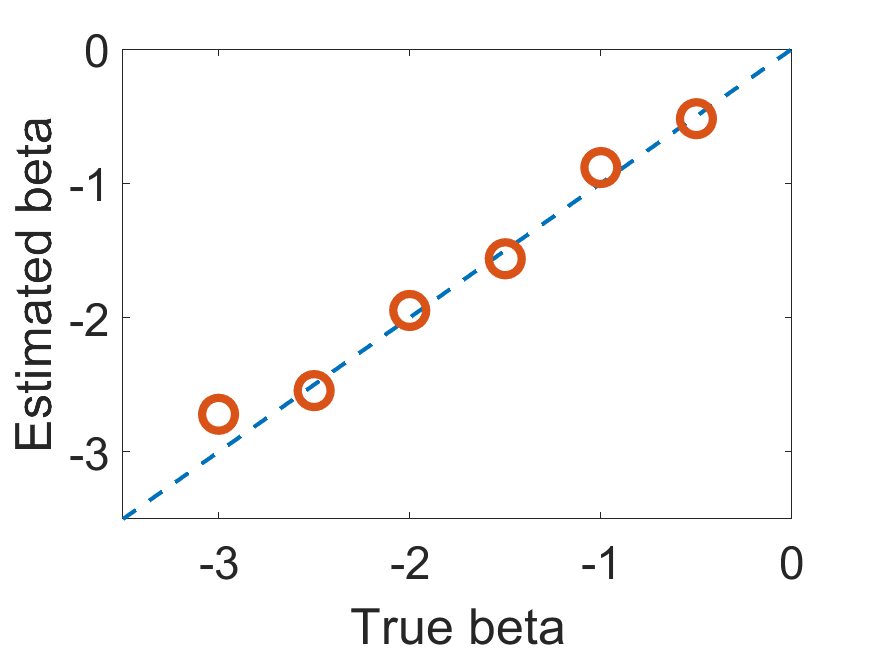}
         \caption{\scriptsize 5 ODs, 100 trips per OD}
         \label{fig:Estimations_5ODs_100TripsPerOD}
    \end{subfigure}
    \begin{subfigure}[b]{0.245 \textwidth}
     \centering
         \includegraphics[width=\textwidth]{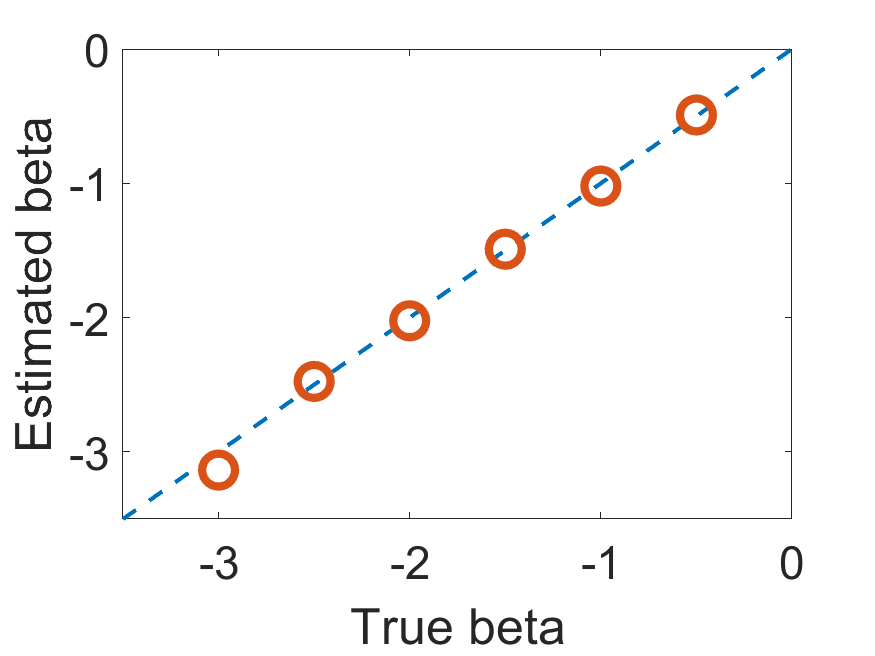}
         \caption{\scriptsize 20 ODs, 100 trips per OD}
         \label{fig:Estimations_20ODs_100TripsPerOD}
    \end{subfigure}
      \begin{subfigure}[b]{0.245 \textwidth}
     \centering
         \includegraphics[width=\textwidth]{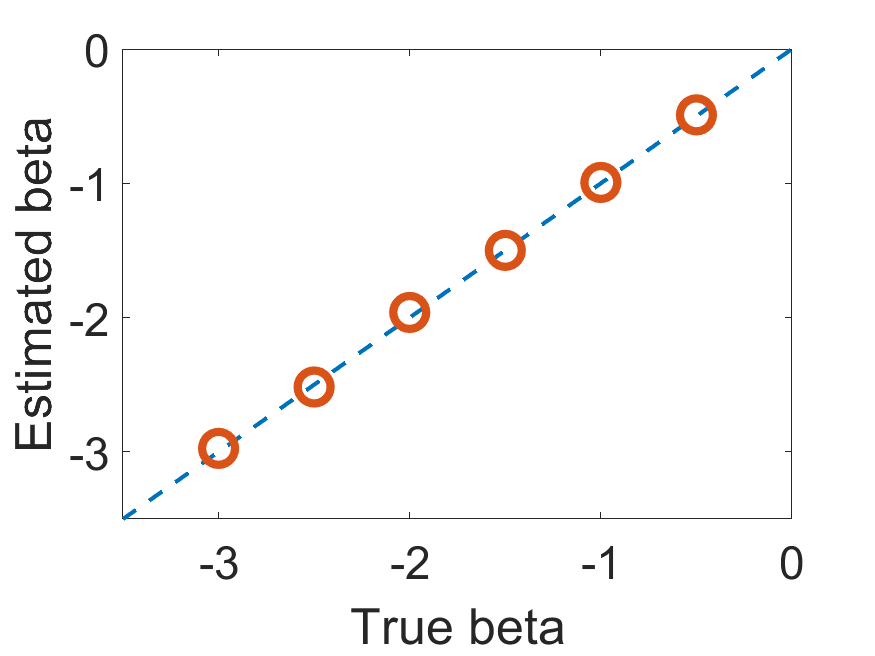}
         \caption{\scriptsize 100 ODs, 100 trips per OD}
         \label{fig:Estimations_100ODs_100TripsPerOD}
    \end{subfigure}
    \\
          \centering
    \begin{subfigure}[b]{0.245 \textwidth}
     \centering
         \includegraphics[width=\textwidth]{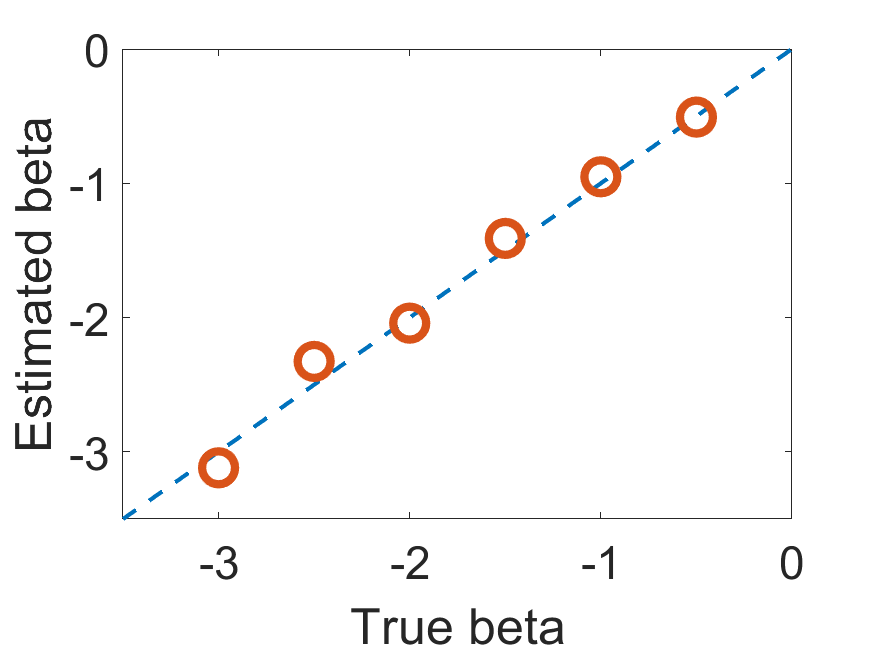}
         \caption{\scriptsize 1 OD, 250 trips per OD}
         \label{fig:Estimations_1ODs_250TripsPerOD}
    \end{subfigure}
  \hfill
      \begin{subfigure}[b]{0.245 \textwidth}
     \centering
         \includegraphics[width=\textwidth]{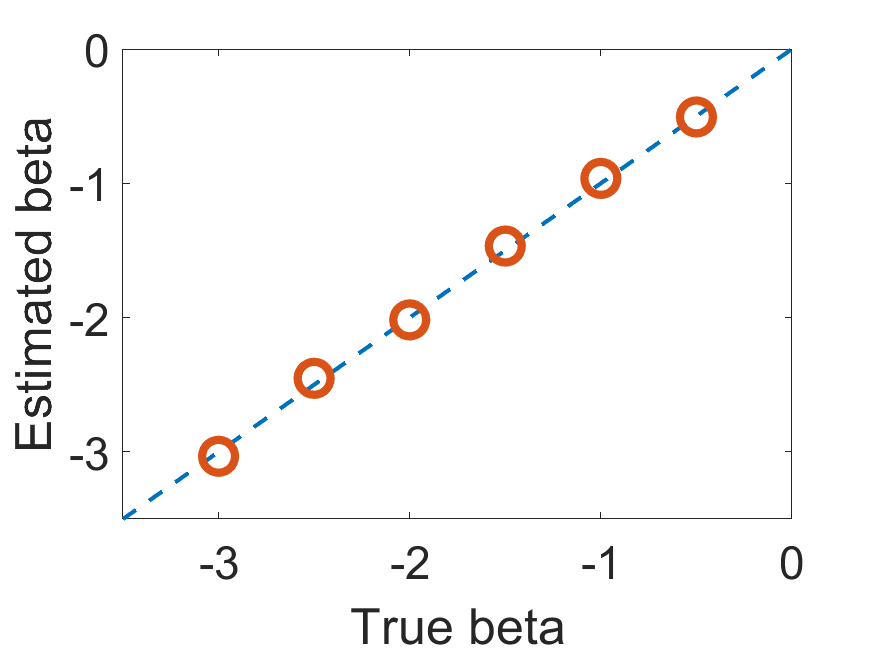}
         \caption{\scriptsize 5 ODs, 250 trips per OD}
         \label{fig:Estimations_5ODs_250TripsPerOD}
    \end{subfigure}
    \begin{subfigure}[b]{0.245 \textwidth}
     \centering
         \includegraphics[width=\textwidth]{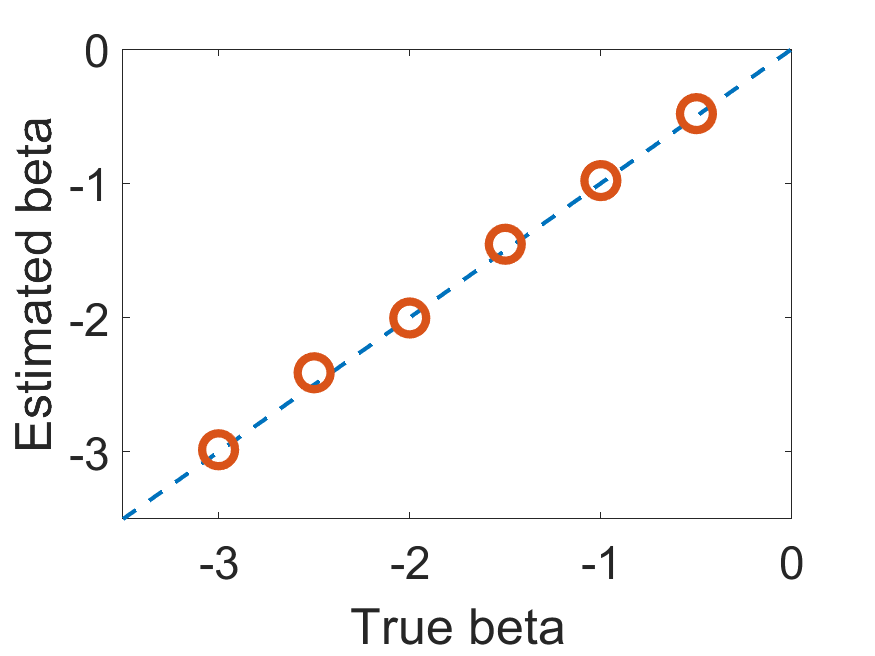}
         \caption{\scriptsize 20 ODs, 250 trips per OD}
         \label{fig:Estimations_20ODs_250TripsPerOD}
    \end{subfigure}
      \begin{subfigure}[b]{0.245 \textwidth}
     \centering
         \includegraphics[width=\textwidth]{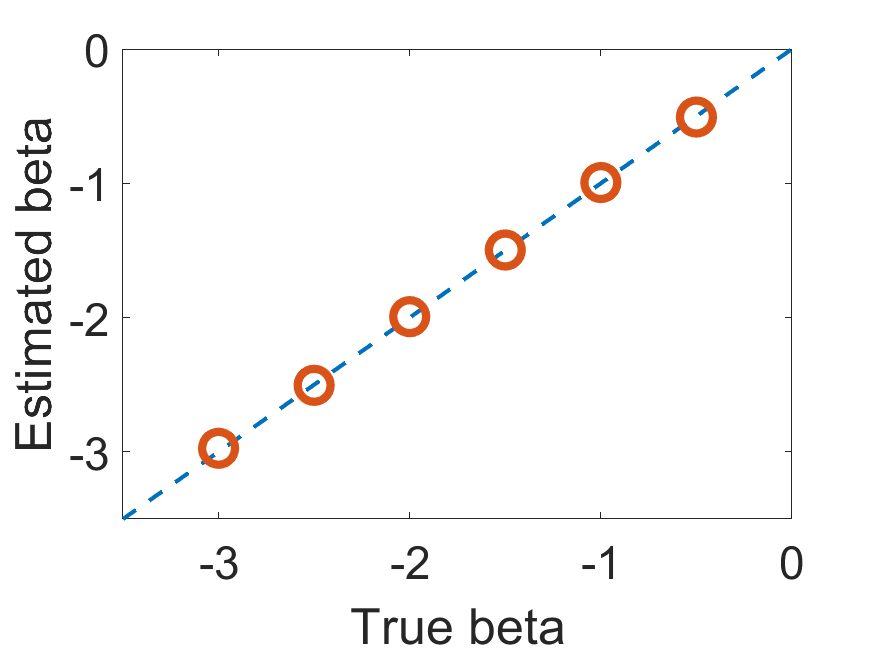}
         \caption{\scriptsize 100 ODs, 250 trips per OD}
         \label{fig:Estimations_100ODs_250TripsPerOD}
    \end{subfigure}
    \\
 
     \centering
    \begin{subfigure}[b]{0.245 \textwidth}
     \centering
         \includegraphics[width=\textwidth]{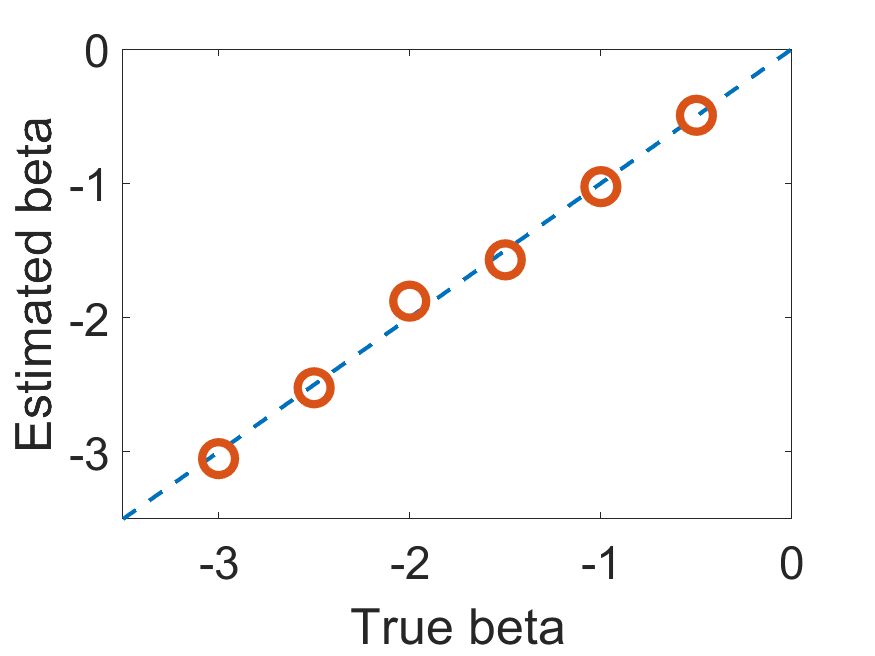}
         \caption{\scriptsize 1 OD, 1,000 trips per OD}
         \label{fig:Estimations_1ODs_1000TripsPerOD}
    \end{subfigure}
  \hfill
      \begin{subfigure}[b]{0.245 \textwidth}
     \centering
         \includegraphics[width=\textwidth]{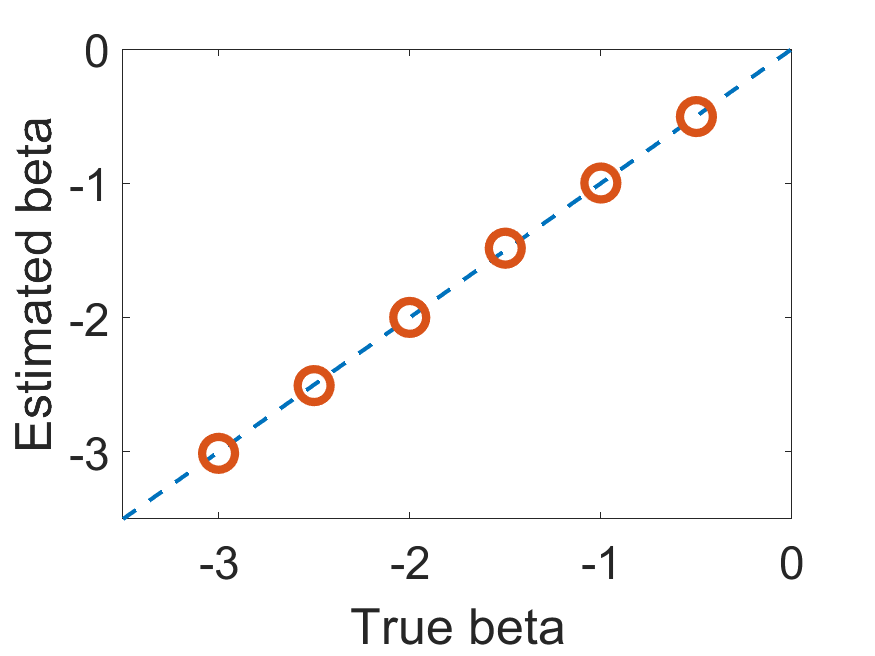}
         \caption{\scriptsize 5 ODs, 1,000 trips per OD}
         \label{fig:Estimations_5ODs_1000TripsPerOD}
    \end{subfigure}
    \begin{subfigure}[b]{0.245 \textwidth}
     \centering
         \includegraphics[width=\textwidth]{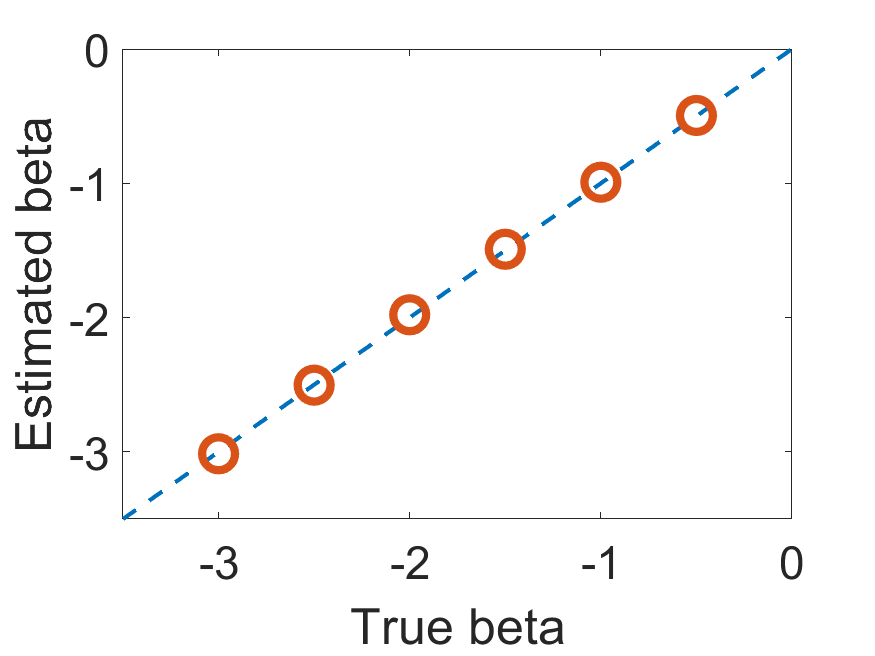}
         \caption{\scriptsize 20 ODs, 1,000 trips per OD}
         \label{fig:Estimations_20ODs_1000TripsPerOD}
    \end{subfigure}
      \begin{subfigure}[b]{0.245 \textwidth}
     \centering
         \includegraphics[width=\textwidth]{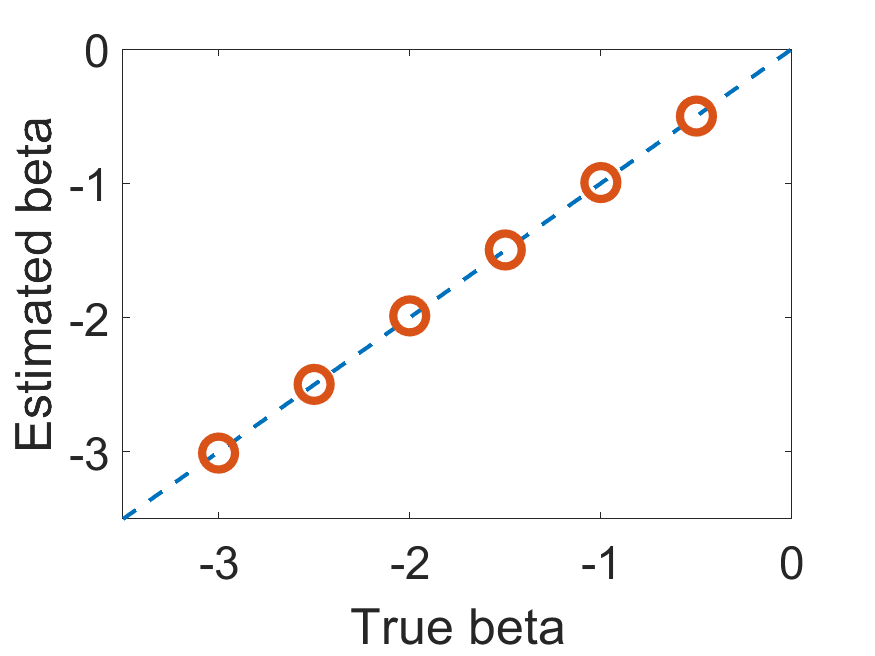}
         \caption{\scriptsize 100 ODs, 1,000 trips per OD}
         \label{fig:Estimations_1O0Ds_1000TripsPerOD}
    \end{subfigure}

    \caption{Parameter estimates against the true parameters for different combinations of the number of trips per OD relation and number of OD relations. }
    \label{fig:estimateBetaSim}
\end{figure}

We have used all observations for all links with positive flow without applying any truncation. This might have led to bias as the function $\ln (1+x)$ that transforms the flows is nonlinear and the flows are estimated with considerable noise. As the results in Figure \ref{fig:estimateBetaSim} show, we find no visible bias when the number of origin-destinations and the number of route choices are sufficiently large. It seems plausible that the reason we do not find any noticeable bias is that the logarithm $\ln(1+x)$ is close to linear for small values of $x$.

Altogether, the results are encouraging. We find that it is possible to recover the true parameter from a sample of routes with good precision. The estimates seem to be biased toward zero, when the dataset is small, but the bias is negligible at moderate dataset sizes.

\bigskip

\section{Empirical application}\label{sec:application}
\subsection{Estimation}
{\tr Our empirical application is based on GPS vehicle trajectory data for the Copenhagen metropolitan area. As described in Appendix  \ref{app:DataProcessing}, we obtain an estimation dataset comprising 8,046 OD combinations and 1,337,096 trips. There are 166.2 trips per OD relation and 153.4 active links on average per OD relation, such that the transformed dataset corresponding to eq. \eqref{eq:OLS} has  1,234,289 observations.  }

To explain the observed route choices, {\tr we specify the utility rate using average observed pace (in minutes per kilometer) for each link},\footnote{A total of 9,464,883 trips were used for determining the empirical average link {\tr paces}.} the number of outlinks, and dummies for road type. Table \ref{tab:Models} reports the estimation results for three different model specifications. Numbers in parentheses are robust standard errors. As can be seen, all parameters are very precisely estimated. {\tr This precision is reassuring, but unsurprising given  the large number of observations entering the regression. }

{\tr Model A has the simplest specification of the utility rate, with just a term for pace. Comparison with the simulations in Figure \ref{fig:EffectOfBetaOnFlows} suggests that the range is reasonable. }

In model B, we extend with a constant for each link having at least two outlinks. This penalizes links that lead into intersections. The {\tr pace} parameter is reduced in absolute value, as expected. {\tr The adjusted R-square increases by 0.006, indicating  preference for model B over model A.}

In model C, we interact the {\tr pace}  variable with link type dummies to estimate {\tr pace} parameters that are specific to each link type. The sign and relative sizes of these parameters make intuitive sense. {\tr The adjusted R-square increases now by 0.04 relative to model B, indicating preference for model C over model B. }

\begin{table}[h]
\begin{tabular}{|l|rrr|}
\hline
 & Model A & Model B & Model C \\ \hline 
$\text{Constant}_{\text{Outlinks}\,\geq\,2}$   & -- & -0.03428 \textit{(0.00030)}  & -0.01546	 \textit{(0.00030)} \\
{\tr Pace} [min/km]& -0.74642	\textit{(0.00171)} & -0.63773		\textit{(0.00225)}  &  --\\
{\tr Pace} [min/km]: & & & \\
\qquad\textit{Motorways} & -- & -- & -0.35011	\textit{(0.00342)}\\
\qquad\textit{Motorway ramps}  & -- & -- &  -0.55419	\textit{(0.00437)}\\
\qquad\textit{Motor traffic roads} & -- & --  &  -0.56233	 \textit{(0.00464)}\\
\qquad\textit{Other national roads} & -- & --  &  -0.59969	 \textit{(0.00297)}\\
\qquad\textit{Urban roads} & -- & --  &  -0.56917	\textit{(0.00219)}\\
\qquad\textit{Rural roads }& -- & --  &  -0.77783	\textit{(0.00308)}\\
\qquad\textit{Smaller roads }& -- & -- & -0.53003	\textit{(0.01101)}\\
\qquad\textit{Other ramps} & -- & -- &  -0.44558	\textit{(0.00448)}\\ \hline
Adjusted $R^2$ & 0.36288 & 0.36855  & 0.40880 \\\hline
\end{tabular}
\caption{Estimation results, based on  1,234,289 observations corresponding to eq. \eqref{eq:OLS}. } \label{tab:Models}
\end{table}

{\tr $R^2$ values around $0.4$ are quite satisfactory considering that we are dealing with cross-sectional data, transformed as described in Section \ref{sec:reg_eq}. We present further validation of model C in Section \ref{sec:validation} below. }

{\tr We have estimated the same models using $F(x)=x^2$ instead of the modified entropy formulation \eqref{eq:modentropy}. The fit was better with the modified entropy formulation in all three cases.

\begin{table}[h]
    \centering
    \tr
    \begin{tabular}{|l|rrr|}
    \hline 
        Model & A & B & C \\ \hline
        Modified entropy & 0.36288 & 0.36855  & 0.40880 \\
        Quadratic & 0.36249 & 0.36695 & 0.40403\\ \hline 
    \end{tabular}
    \tr
    \caption{\tr Adjusted $R^2$ values for modified entropy and quadratic perturbation functions.}
    \label{tab:my_label}
\end{table}
}

\subsection{Validation}\label{sec:validation}

We go on to validate the prediction of model C against our data. For each OD combination in the data, we compute the predicted flow vector using the estimated parameters. We compare this prediction to the estimation dataset to ascertain the ability of the model to reproduce these data. 
The main text shows validation results using the full sample. We have a large dataset and a model with few parameters, so the results are not likely to be very vulnerable to over-fitting. However, we test out-of-sample predictions in Appendix \ref{app:DataProcessing:Dividing}, where we have split the data at random by OD combination into two equal-sized datasets. We have estimated the model for each sample split and applied each set of  estimated parameters to make predictions for the other half of the sample. The parameter estimates for the two sample splits found in Appendix \ref{app:OutOfSample} are very similar and the out-of-sample prediction test did not reveal any issues.

Figure \ref{fig:ValTotalLinkFlows} shows the main validation result, comparing observed link flows to the predicted flows. {\tr For each link, the observed link flow is the number of observed vehicles on that link in the observation dataset across all OD combinations, whereas the predicted link flows are found by multiplying the number of observed trips for each OD combination with the corresponding predicted link flow for a representative traveler for that particular OD combination before summing across all OD combinations. We observe that predictions are generally close to the $45^\circ$ line. A nonparametric fit with 95\% confidence bands suggests some systematic deviations.  The adjusted $R^2$ between the observed and predicted link flows is 0.9356.\footnote{Defined as $R^2_\text{adj}= 1 -  \left(
1-\dfrac{\sum_{i=1}^N \left( \hat{x}_i - x_i  \right)^2}{\sum_{i=1}^N\left(x_i - \bar{x}\right)^2} \right) \frac{1-N}{1-p-N}$, with $x_i, i = 1, 2, ..., N$ being the observed link flows, $\hat{x}_i, i = 1, 2, ..., N$ being the predicted link flows, $\bar{x}= \sum_{i=1}^N x_i$ being the mean of the $N = 30,772$ observed link flows, and $p=9$ being the number of estimated parameters.}   We do find the fit to be very satisfactory, as we have employed no calibration: the predictions are driven solely by the nine estimated parameters in Table \ref{tab:Models}. Our general conclusion is that the model is able to meaningfully predict the observed flows.}

\begin{figure}[H]
    \centering
    \includegraphics[width=0.7\textwidth]{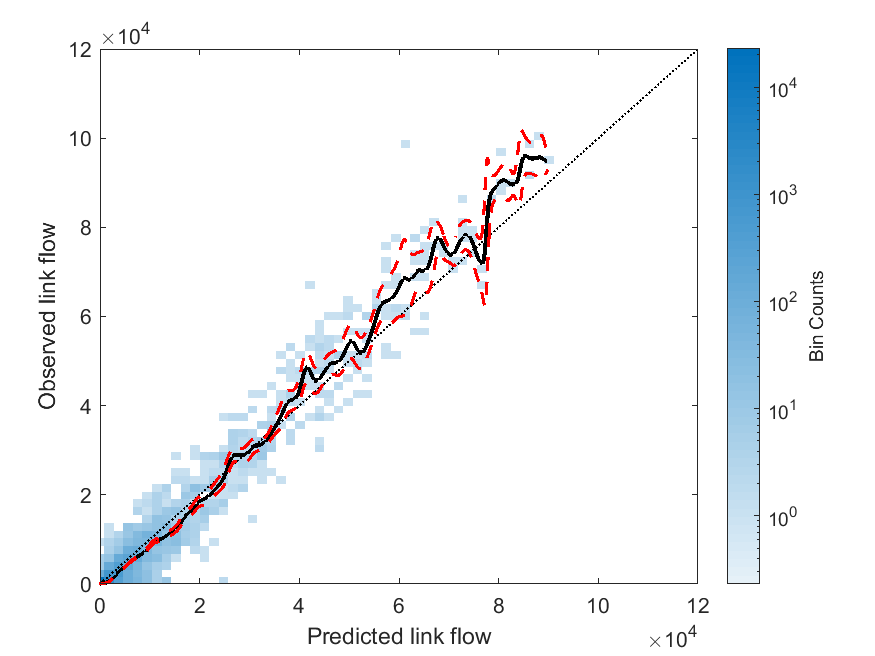} 
    \caption{Observed vs. predicted total link flows when summing across all OD combinations.}
    \label{fig:ValTotalLinkFlows}
\end{figure}

\begin{figure}[H]
    \centering
    \includegraphics[width= \textwidth]{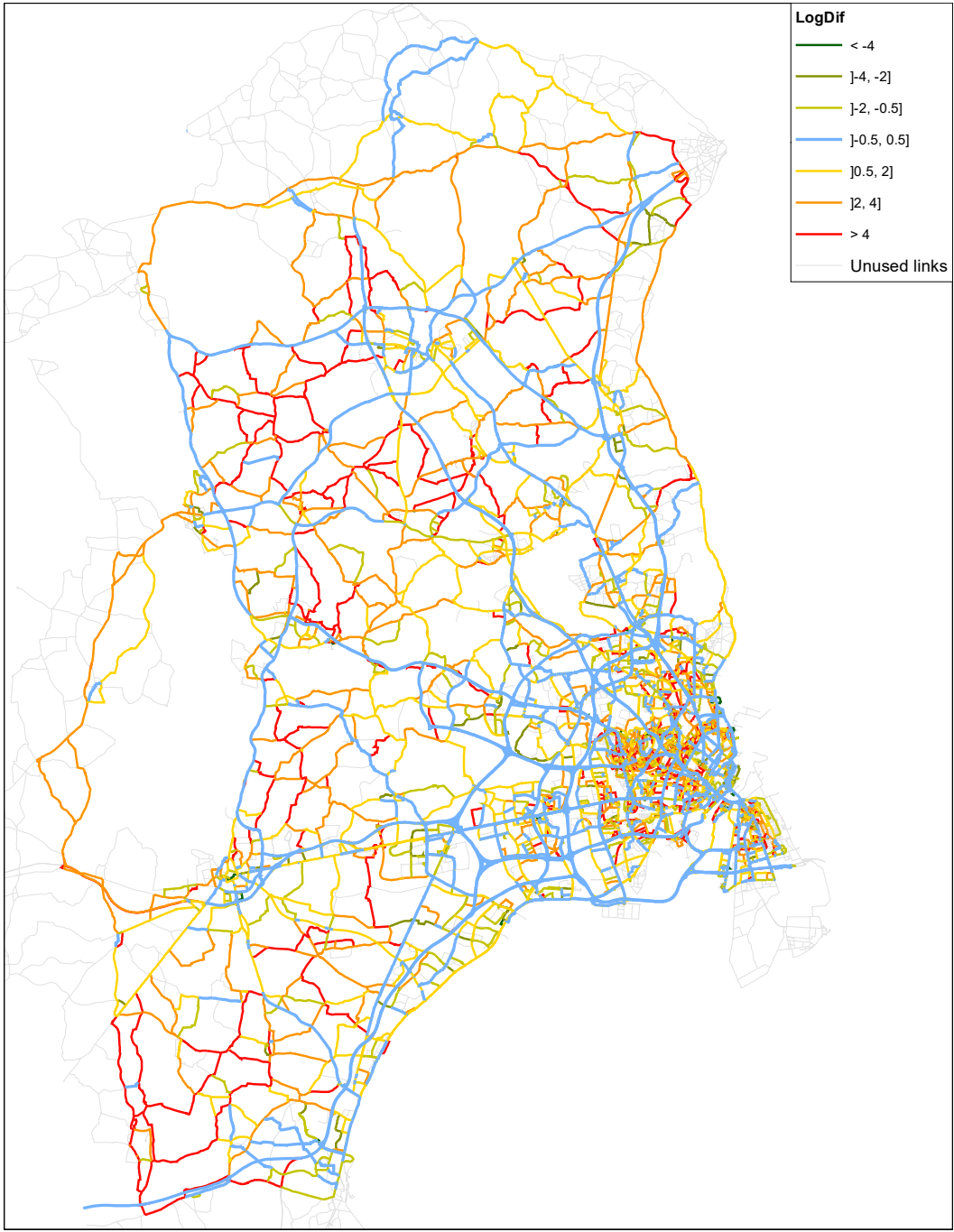}
    \caption{Regional difference map. Differences given as $\ln(q^\text{predicted} + 1)- \ln (q^\text{observed}+1)$.}
    \label{fig:ValMap_Regional}
\end{figure}

\begin{figure}[H]
    \centering
    \includegraphics[width= \textwidth]{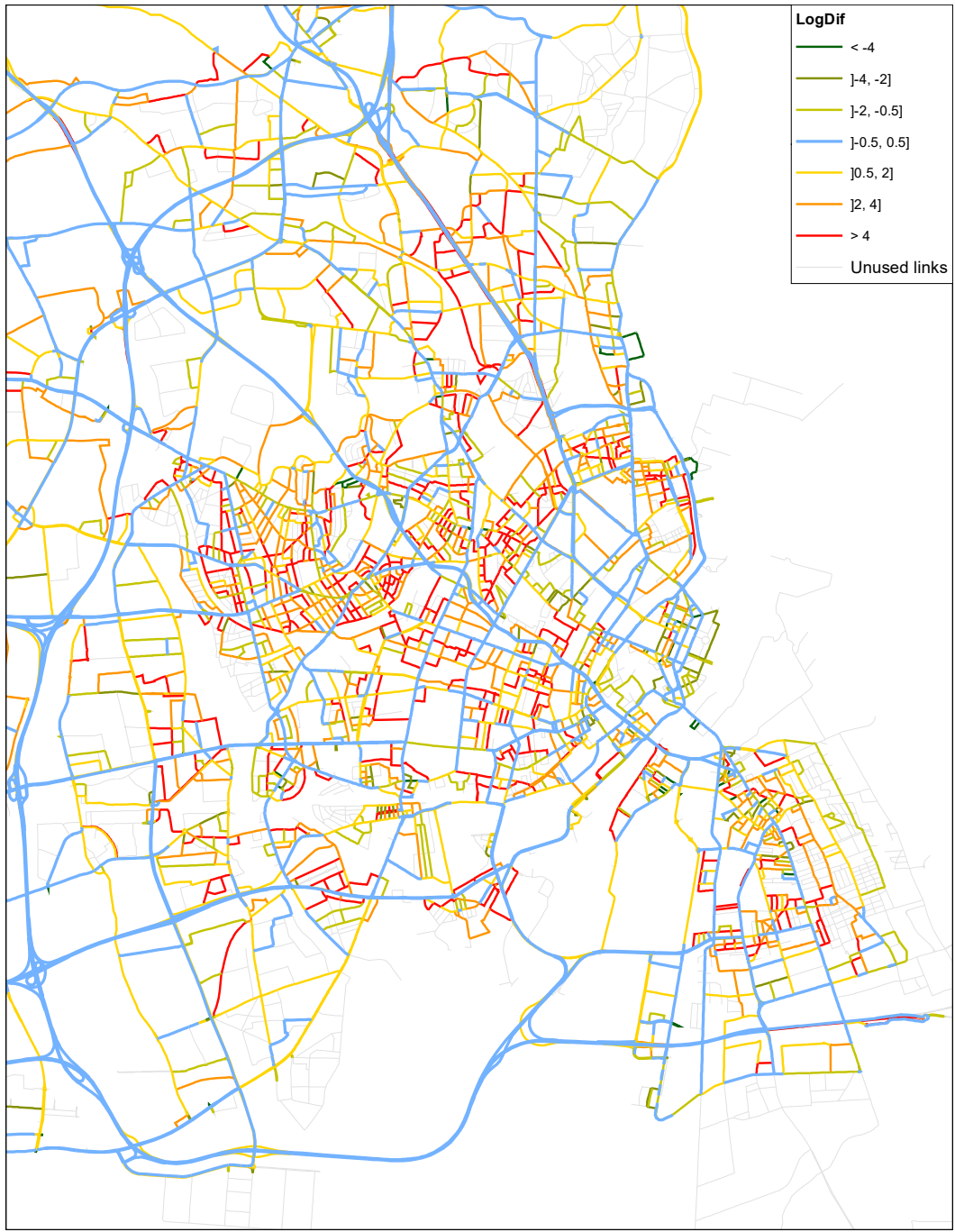}
    \caption{City-wide difference map. Differences given as $\ln(q^\text{predicted} + 1)- \ln (q^\text{observed} + 1$).}
    \label{fig:ValMap_City}
\end{figure}

Figures \ref{fig:ValMap_Regional} and \ref{fig:ValMap_City} visualize the comparison between predicted and observed flows on a map of the network. Blue links are well predicted, the model predicts too much traffic on red links and too little on green links compared with the observed flows. We observe that the main road network is generally blue on the figure, whereas red links tend to be minor roads. We find this result to be quite {\tr reasonable}, indicating {\tr again} that the estimated model is able to predict a large share of the variation in the observed link flows.

\begin{figure}[H]
    \centering
    \includegraphics[width=0.7 \textwidth]{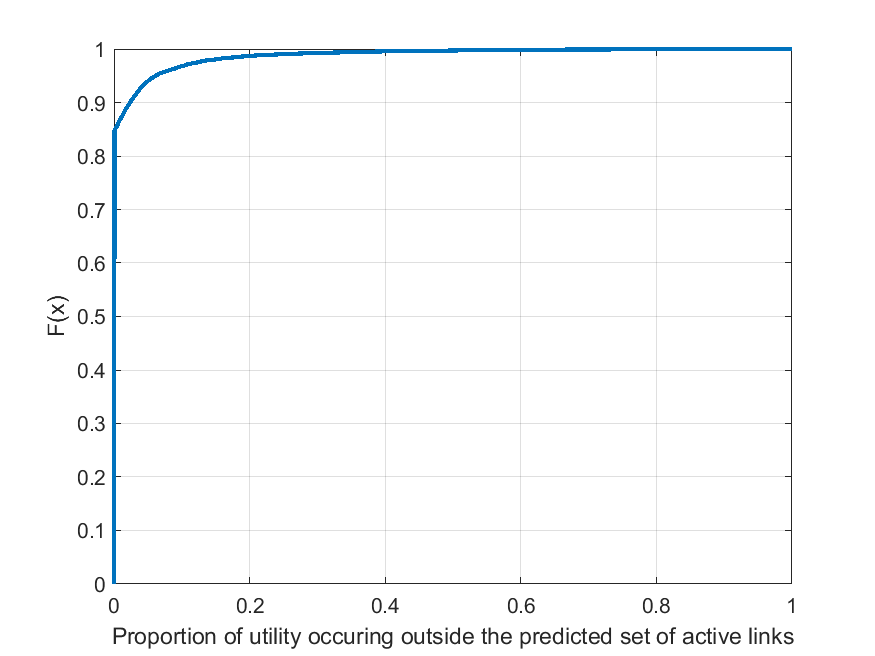}
    \caption{Cumulative distribution function of the proportion of utility of the observed trips outside the predicted set of active links.}
    \label{fig:ValUnused}
\end{figure}

An issue that has troubled modeling approaches that rely on a predefined consideration set is that the observed routes may not be covered by the consideration set. In contrast, the perturbed utility route choice model takes all potential routes into account and predicts a set of active links. {\tr
The model predicts 12,133 unused links with a total length of 4,253 km. The observed  data comprises 12,063 unused links with a total length of  4,719 km. The overlap between the two sets of unused links is 79\%, showing that the vast majority of the unused links is correctly predicted.}

{\tr In addition, we analyze how  well the predicted sets of active links, one set for each OD combination, cover the chosen  routes.} Figure \ref{fig:ValUnused} shows the cumulative distribution of the proportion of utility of the chosen routes that is outside the predicted set of active links. Hence, a value of zero is perfect and small values are good. We find that about 85\% of the observed trips are completely inside the predicted set of active links, and almost all observed trips have less than 20\% of utility outside the predicted set of active links.

\section{Conclusion} \label{sec:conclusion}
This paper has opened the door to a new way to model route choice. The perturbed utility route choice model can be estimated using just linear regression, {\tr  it requires no choice set generation, }it is straightforward to use for assignment, and it produces reasonable substitution patterns. We have estimated and applied the model to a large dataset covering the Copenhagen metropolitan area and found that the model is able to predict much of the variation in observed flows, thereby demonstrating the feasibility of applying the model in practice. 

We do not claim to have found the perfect model. However, we do hope to have convinced the readers that the perturbed utility route choice model is a worthwhile alternative to previous existing approaches. There are many ways to go forward from here.  

We have ignored sampling error in the observation of flows. Our simulations suggest that this is not a problem when the number of observations is sufficiently large. Further work might incorporate weighting to reflect the number of observations used for estimating flows. It would also be of interest to be able to estimate the model with individual-level data, perhaps adapting the iteratively reweighted least squares algorithm of \citet{Nielsen2021}. {\tr This might be relevant when using the perturbed utility route choice model with smaller datasets, where sampling error might be a more prominent issue.} 

We have worked with a static and deterministic network. It might be possible to adapt the perturbed utility route choice model to more complex settings.

{\tr An interesting question is the shape of the perturbation function. There is considerable freedom within the requirement imposed on $F$, and there are many other options in addition to the two we have tested. The effect of using other perturbation functions remains an open issue.}  

Another direction that could be explored is to add cross-link interactions as a means to control the substitution patterns of the model. One way to go might be along the lines of the nested recursive logit model \citep{Mai2015}. 

Finally, we note that our model seems to fit well into the equilibrium assignment in \citet{Oyama2020}. This research direction also seems appealing.

\pagebreak
\appendix{}\label{appendix}

\section{Data processing} \label{app:DataProcessing}

The raw data comprise 472,801,802 GPS points corresponding to 9,464,883 trips. The data cover a three-month period from June 1, 2019 to August 31, 2019.

\subsection{Map-matching} \label{app:DataProcessing:MM}

The data were map-matched to the network using a branch-and-bound-based algorithm proposed by \cite{Nielsenetal2007}.  This led to 8,039,296 map-matched trips. Some trips were only partially map-matched; we retained 5,320,215  trips for which the whole trajectory could be matched successfully. 

\subsection{Trimming trips} \label{app:DataProcessing:Trimming}

To ensure that we have sufficient trips with OD combination in common, we devised an algorithm to trim the start of trips  such that they begin in a smaller set of nodes in the network. We used the same algorithm to trim the end of trips to arrive at a smaller set of common destination nodes.

The algorithm is based on a sample of trips $\mathcal{T}$, where each  trip $t \in \mathcal{T}$ is the sequence of nodes visited $t=(v^t_1,\ldots,v^t_{n_t})$. 
Denote by $\mathcal{T}_v$ the set of trips that passes through node $v$, and let $k_\ast^t$ denote the first occurrence of a selected node for trip $t \in \mathcal{T}$. The algorithm iteratively trims trips in a given set of trips until the number of unique origin nodes is reduced to a prespecified number  $N_O$. It comprises the following steps. 

\begin{enumerate}
    \item Initialize the set of chosen origin nodes, $\mathcal{V}_O$, as the empty set, $\mathcal{V}_O = \emptyset$, and set $k_\ast^t = n^t$ for all $t \in \mathcal{T}$.
    \item Calculate the trip score $S^t_{v^t_k}$ for every node $v^t_k, k = 1,\ldots,n_t$ of every trip $t \in \mathcal{T}$ as the number of remaining nodes of that trip: $S_{v^t_k} \gets k_\ast^t - k.$
   \item Calculate the score of each node  $v \in \mathcal{V}$ as the sum of the trip scores across trips: $S_v \gets  \sum\limits_{t \in \mathcal{T}_v}S^t_v$.
   \item Add the highest scoring node $v^\ast = \argmax\limits_{v \in \mathcal{V}} S_v$ to the set of selected nodes: $\mathcal{V}_O \gets \mathcal{V}_O \cup v^\ast $.
   \item For each trip $t \in \mathcal{T}$ for which $v^\ast \in t$, (potentially) update the first occurrence of a selected node by $k^t_\ast \gets \min k : v^t_k \in \mathcal{V}_O$.
   \item Repeat steps 2--5 until $|\mathcal{V}_O| = N_O$.
    \item Return $\mathcal{V}_O$.
\end{enumerate}

We ran the algorithm to trim the data to 100 origins and 100 destinations to obtain the sets of valid origins and destinations, $\mathcal{V}_O$ and $\mathcal{V}_D$, respectively. 
Figure \ref{fig:ODs} shows the resulting origins and destinations.

This allows trimming each trip $t \in \mathcal{T}$ so that it begins the first time it visits a node $v \in \mathcal{V}_O$, and ends the last time it visits a node $v \in \mathcal{V}_D$. 
A trip has to visit its first valid origin node before reaching its last valid destination node, otherwise the trip is discarded. 
Likewise, trips that do not pass through a valid origin node and a valid destination node are also discarded. 
We retained 1,680,667 trimmed trips that visit their first valid origin node before reaching their final valid destination node. 




\begin{figure}[H]
    \centering
     \includegraphics[width=0.7\textwidth]{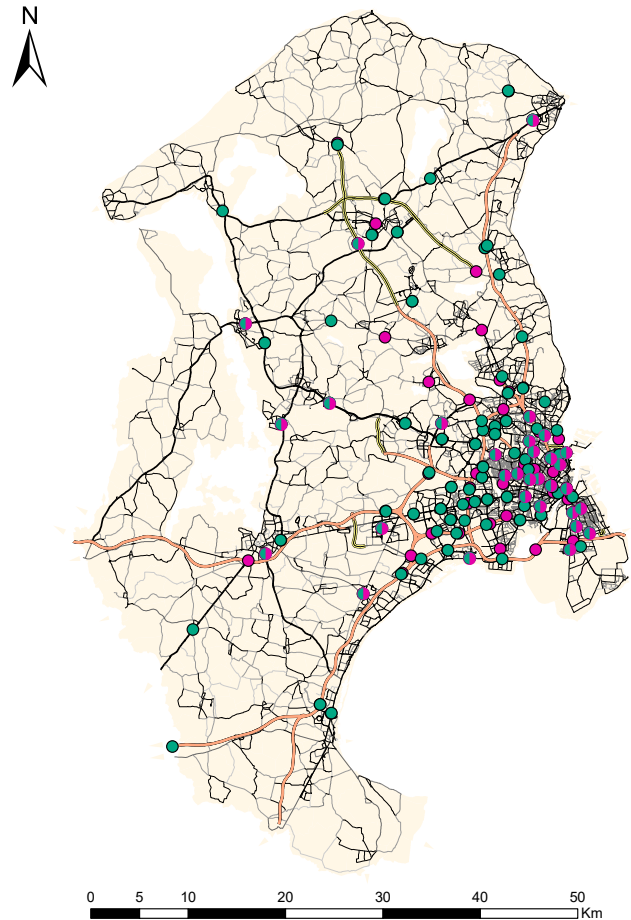} 
    \caption{"Pseudo" origin (green) and destination (red) nodes.}
    \label{fig:ODs}
\end{figure}

\subsection{Deleting non-sensical observations} \label{app:DataProcessing:NonSense}

The 1,680,667 observations with valid OD sections covered a total of 9,284 different OD combinations.
We discarded 1,238 OD combinations  for which all trips used the same links; 187,790 trips were thus discarded.

Some observed trips were non-sensical, perhaps owing to detours for intermediate trip purposes. To identify such trips, we ran a prediction of the model for each OD combination, using a {\tr pace} parameter of $-0.3$, which is much lower than the estimated parameter in model A. The set of active links predicted as active with that single parameter is then quite large and likely to comprise most reasonable routes. 

We then used the selection criterion that observed routes should have at least 95\% of their utility (given in this case simply by {\tr pace}) inside the active set. This criterion led us to discard 155,781 trips. 

In the end, we arrived at an estimation dataset comprising 8,046 active OD combinations and 1,337,096 trips.

\section{Estimation results for sample splits} \label{app:OutOfSample}
 \label{app:DataProcessing:Dividing}

To perform out-of-sample estimations, the 8,046 OD combinations were randomly split in two halves. 
Both contain 4,023 OD combinations and the number of trips is about the same;  668,330 trips in sample split 1 and 668,766 trips in sample split 2.

Tables \ref{tab:ModelsA} and \ref{tab:ModelsB} present the estimation results for the two sample splits. The parameter estimates are very similar.

\begin{table}[H]
\begin{tabular}{|l|rrr|}
\hline
 & Model A & Model B & Model C \\ \hline 
$\text{Constant}_{\text{Outlinks} \geq 2}$   & -- & -0.03423 \textit{(0.00043)}  & -0.01512	 \textit{(0.00044)} \\
{\tr Pace} [min/km]& -0.75163	\textit{(0.00246)} & -0.64228		\textit{(0.00325)}  &  --\\
{\tr Pace} [min/km]: &&& \\
\qquad\textit{Motorways} & -- & -- & -0.35801	\textit{(0.00492)}\\
\qquad\textit{Motorway ramps}  & -- & -- &  -0.55902	\textit{(0.00621)}\\
\qquad\textit{Motor traffic roads} & -- & --  &  -0.58478	 \textit{(0.00638)}\\
\qquad\textit{Other national roads} & -- & --  &  -0.60969	 \textit{(0.00428)}\\
\qquad\textit{Urban roads} & -- & --  &  -0.57789	\textit{(0.00322)}\\
\qquad\textit{Rural roads }& -- & --  &  -0.77801	\textit{(0.00438)}\\
\qquad\textit{Smaller roads }& -- & -- & -0.54239	\textit{(0.01211)}\\
\qquad\textit{Other ramps} & -- & -- &  -0.46494	\textit{(0.00626)}\\ \hline
Adjusted $R^2$ & 0.36394 & 0.36956  & 0.40945 \\ \hline 
\end{tabular}
\caption{Estimation results for sample split 1, based on  615,776  observations corresponding to eq. \eqref{eq:OLS}.} \label{tab:ModelsA}
\end{table}

\begin{table}[H]
\begin{tabular}{|l|rrr|}
\hline
 & Model A & Model B & Model C \\ \hline 
$\text{Constant}_{\text{Outlinks} \geq 2}$   & -- & -0.03428 \textit{(0.00030)}  & -0.01577	 \textit{(0.00042)} \\
{\tr Pace} [min/km] & -0.74131	\textit{(0.00238)} & -0.63332		\textit{(0.00312)}  &  --\\
{\tr Pace} [min/km]: &&& \\
\qquad\textit{Motorways} & -- & -- & -0.34169	\textit{(0.00475)}\\
\qquad\textit{Motorway ramps}  & -- & -- &  -0.54917	\textit{(0.00614)}\\
\qquad\textit{Motor traffic roads} & -- & --  &  -0.53683	 \textit{(0.00676)}\\
\qquad\textit{Other national roads} & -- & --  &  0.58954	 \textit{(0.00410)}\\
\qquad\textit{Urban roads} & -- & --  &  -0.56045	\textit{(0.00297)}\\
\textit{Rural roads }& -- & --  &  -0.77761	\textit{(0.00435)}\\
\qquad\textit{Smaller roads }& -- & -- & -0.51859	\textit{(0.01855)}\\
\qquad\textit{Other ramps} & -- & -- &  -0.42436	\textit{(0.00642)}\\ \hline
Adjusted $R^2$ & 0.36185 & 0.36755  & 0.40832 \\\hline 
\end{tabular}
\caption{Estimation results for sample split 2, based on  618,513 observations corresponding to eq. \eqref{eq:OLS}.} \label{tab:ModelsB}
\end{table}

We performed two sets of out-of-sample predictions with the models, using parameters estimated on one sample split in the prediction for the other sample split. Figures \ref{fig:AValTotalLinkFlows} and \ref{fig:BValTotalLinkFlows} show scatter plots of observed versus predicted link flows. {\tr $R^2_\text{adj}=0.9353$ for sample split 1, and $R^2_\text{adj} = 0.9314$ for sample split 2.}

\begin{figure}[H]
    \centering
    \includegraphics[width=0.7 \textwidth]{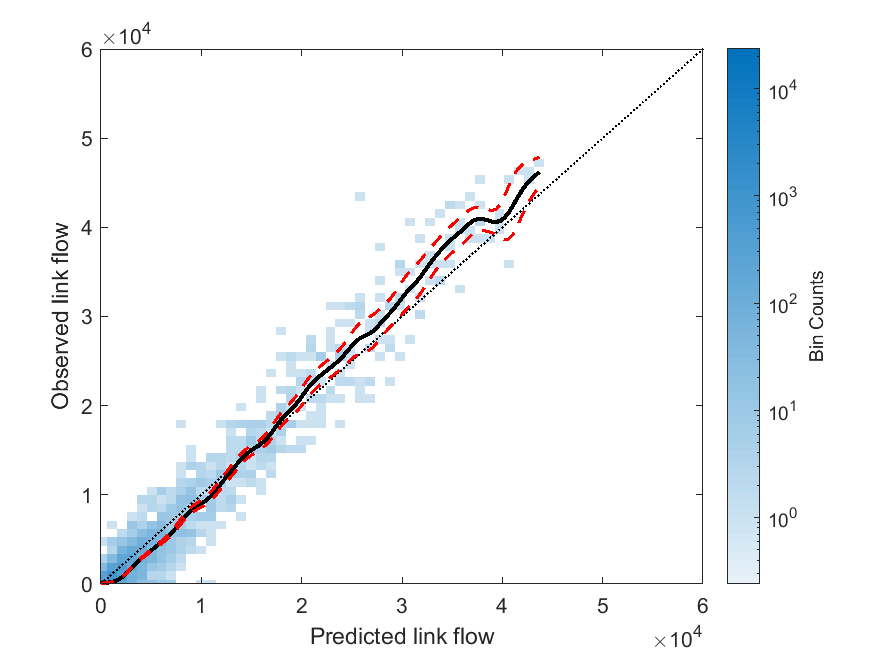}
    \caption{Observed vs. predicted total link flows when summing across all OD combinations (sample split 2 using parameters estimated on sample split 1).}
    \label{fig:AValTotalLinkFlows}
\end{figure}

\begin{figure}[H]
    \centering
    \includegraphics[width=0.7 \textwidth]{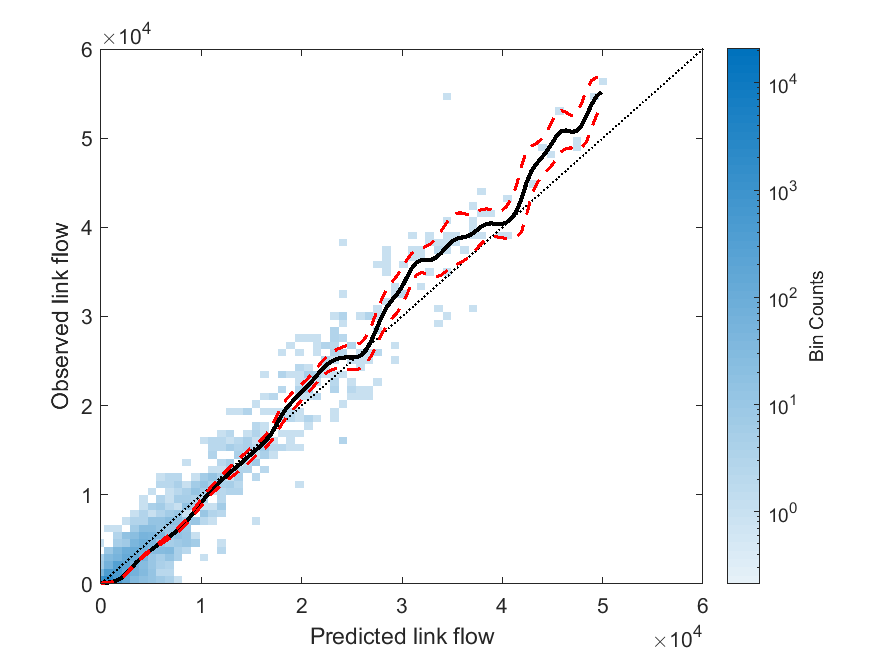}
    \caption{Observed vs. predicted total link flows when summing across all OD combinations (sample split 1 using parameters estimated on sample split 2).}
    \label{fig:BValTotalLinkFlows}
\end{figure}

\pagebreak
\bibliographystyle{ecta}
\bibliography{MFreferences,madspReferences,tkraRefs}

\begin{thebibliography}{33}
\newcommand{\enquote}[1]{``#1''}
\expandafter\ifx\csname natexlab\endcsname\relax\def\natexlab#1{#1}\fi

\bibitem[\protect\citeauthoryear{Allen and Rehbeck}{Allen and
  Rehbeck}{2019{\natexlab{a}}}]{Allen2019b}
\textsc{Allen, R. and J.~Rehbeck} (2019{\natexlab{a}}): \enquote{{Hicksian
  complementarity and perturbed utility models},} .

\bibitem[\protect\citeauthoryear{Allen and Rehbeck}{Allen and
  Rehbeck}{2019{\natexlab{b}}}]{Allen2019}
---\hspace{-.1pt}---\hspace{-.1pt}--- (2019{\natexlab{b}}):
  \enquote{{Identification With Additively Separable Heterogeneity},}
  \emph{Econometrica}, 87, 1021--1054.

\bibitem[\protect\citeauthoryear{Ataei}{Ataei}{2014}]{Ataei2014}
\textsc{Ataei, A.} (2014): \enquote{{Improved Qrginv Algorithm for Computing
  Moore-Penrose Inverse Matrices},} \emph{ISRN Applied Mathematics}, 2014,
  1--5.

\bibitem[\protect\citeauthoryear{Baillon and Cominetti}{Baillon and
  Cominetti}{2008}]{Baillon2008}
\textsc{Baillon, J.~B. and R.~Cominetti} (2008): \enquote{{Markovian traffic
  equilibrium},} \emph{Mathematical Programming}, 111, 33--56.

\bibitem[\protect\citeauthoryear{Bell}{Bell}{1995}]{Bell1995}
\textsc{Bell, M.~G.} (1995): \enquote{{Alternatives to Dial's logit assignment
  algorithm},} \emph{Transportation Research Part B}, 29, 287--295.

\bibitem[\protect\citeauthoryear{Ben-Akiva and Bierlaire}{Ben-Akiva and
  Bierlaire}{1999}]{Ben-Akiva1999PSL}
\textsc{Ben-Akiva, M. and M.~Bierlaire} (1999): \emph{Discrete Choice Methods
  and their Applications to Short Term Travel Decisions}, Boston, MA: Springer
  US, 5--33.

\bibitem[\protect\citeauthoryear{Byrd, Hribar, and Nocedal}{Byrd
  et~al.}{1999}]{Byrd1999}
\textsc{Byrd, R.~H., M.~E. Hribar, and J.~Nocedal} (1999): \enquote{{An
  Interior Point Algorithm for Large-Scale Nonlinear Programming},} \emph{SIAM
  Journal on Optimization}, 9, 877--900.

\bibitem[\protect\citeauthoryear{Dial}{Dial}{1971}]{Dial1971}
\textsc{Dial, R.~B.} (1971): \enquote{{A probabilistic multipath traffic
  assignment algorithm which obviates path enumeration},} \emph{Transportation
  Research}, 5, 83--111.

\bibitem[\protect\citeauthoryear{Duncan, Watling, Connors, Rasmussen, and
  Nielsen}{Duncan et~al.}{2020}]{Duncan2020}
\textsc{Duncan, L.~C., D.~P. Watling, R.~D. Connors, T.~K. Rasmussen, and O.~A.
  Nielsen} (2020): \enquote{Path Size Logit route choice models: Issues with
  current models, a new internally consistent approach, and parameter
  estimation on a large-scale network with GPS data,} \emph{Transportation
  Research Part B: Methodological}, 135, 1--40.

\bibitem[\protect\citeauthoryear{Duncan, Watling, Connors, Rasmussen, and
  Nielsen}{Duncan et~al.}{2021}]{Duncan2021}
---\hspace{-.1pt}---\hspace{-.1pt}--- (2021): \enquote{A bounded path size
  route choice model excluding unrealistic routes: formulation and estimation
  from a large-scale GPS study,} \emph{Transportmetrica A: Transport Science},
  1--59.

\bibitem[\protect\citeauthoryear{Feng, Li, and Wang}{Feng
  et~al.}{2018}]{Feng2018}
\textsc{Feng, G., X.~Li, and Z.~Wang} (2018): \enquote{{On substitutability and
  complementarity in discrete choice models},} \emph{Operations Research
  Letters}, 46, 141--146.

\bibitem[\protect\citeauthoryear{Fosgerau, Frejinger, and Karlstrom}{Fosgerau
  et~al.}{2013}]{Fosgerau2013e}
\textsc{Fosgerau, M., E.~Frejinger, and A.~Karlstrom} (2013): \enquote{{A link
  based network route choice model with unrestricted choice set},}
  \emph{Transportation Research Part B: Methodological}, 56, 70--80.

\bibitem[\protect\citeauthoryear{Fosgerau and McFadden}{Fosgerau and
  McFadden}{2012}]{McFadden2012}
\textsc{Fosgerau, M. and D.~L. McFadden} (2012): \enquote{{A theory of the
  perturbed consumer with general budgets},} \emph{NBER Working Paper}, 1--27.

\bibitem[\protect\citeauthoryear{Fudenberg, Iijima, and Strzalecki}{Fudenberg
  et~al.}{2015}]{Fudenberg2015}
\textsc{Fudenberg, D., R.~Iijima, and T.~Strzalecki} (2015):
  \enquote{{Stochastic Choice and Revealed Perturbed Utility},}
  \emph{Econometrica}, 83, 2371--2409.

\bibitem[\protect\citeauthoryear{Hofbauer and Sandholm}{Hofbauer and
  Sandholm}{2002}]{Hofbauer2002}
\textsc{Hofbauer, J. and W.~H. Sandholm} (2002): \enquote{{On the global
  convergence of stochastic fictitious play},} \emph{Econometrica}, 70,
  2265--2294.

\bibitem[\protect\citeauthoryear{Kjems and Paag}{Kjems and
  Paag}{2019}]{Kjems2019}
\textsc{Kjems, S. and H.~Paag} (2019): \enquote{{COMPASS: Ny trafikmodel for
  hovedstadsomr{\aa}det},} \emph{Proceedings from the Annual Transport
  Conference at Aalborg University}.

\bibitem[\protect\citeauthoryear{Mai}{Mai}{2016}]{Mai2016}
\textsc{Mai, T.} (2016): \enquote{A method of integrating correlation
  structures for a generalized recursive route choice model,}
  \emph{Transportation Research Part B: Methodological}, 93, 146--161.

\bibitem[\protect\citeauthoryear{Mai, Fosgerau, and Frejinger}{Mai
  et~al.}{2015{\natexlab{a}}}]{Mai2015}
\textsc{Mai, T., M.~Fosgerau, and E.~Frejinger} (2015{\natexlab{a}}):
  \enquote{{A nested recursive logit model for route choice analysis},}
  \emph{Transportation Research Part B: Methodological}, 75, 100--112.

\bibitem[\protect\citeauthoryear{Mai, Frejinger, and Bastin}{Mai
  et~al.}{2015{\natexlab{b}}}]{Mai2015a}
\textsc{Mai, T., E.~Frejinger, and F.~Bastin} (2015{\natexlab{b}}): \enquote{{A
  Dynamic Programming Approach for Quickly Estimating Large Scale MEV Models},}
  \emph{Transportation Research Part B: Methodological}.

\bibitem[\protect\citeauthoryear{{MATLAB}}{{MATLAB}}{2010}]{MATLAB2010}
\textsc{{MATLAB}} (2010): \emph{{MATLAB:2010}}, Natick, Massachusetts: The
  MathWorks Inc.

\bibitem[\protect\citeauthoryear{McFadden}{McFadden}{1981}]{McFadden1981}
\textsc{McFadden, D.} (1981): \enquote{{Econometric Models of Probabilistic
  Choice},} in \emph{Structural Analysis of Discrete Data with Econometric
  Applications}, ed. by C.~Manski and D.~McFadden, Cambridge, MA, USA: MIT
  Press, 198--272.

\bibitem[\protect\citeauthoryear{McFadden}{McFadden}{1973}]{McFadden1974a}
\textsc{McFadden, D.~L.} (1973): \enquote{{Conditional logit analysis of
  qualitative choice behavior},} in \emph{Frontiers in Econometrics}, New York:
  Academic Press, 105–142.

\bibitem[\protect\citeauthoryear{Nielsen, Fosgerau, and Kristensen}{Nielsen
  et~al.}{2021}]{Nielsen2021}
\textsc{Nielsen, N., M.~Fosgerau, and D.~Kristensen} (2021):
  \enquote{{Estimation of perturbed utility models using demand inversion},}
  \emph{Working Paper}.

\bibitem[\protect\citeauthoryear{Nielsen, W{\"u}rtz, and J{\o}rgensen}{Nielsen
  et~al.}{2007}]{Nielsenetal2007}
\textsc{Nielsen, O., C.~W{\"u}rtz, and R.~J{\o}rgensen} (2007):
  \enquote{Improved map-matching algorithms for GPS-data - - Methodology and
  test on data from The AKTA roadpricing experiment in Copenhagen.} European
  ITS Conference ; Conference date: 18-06-2007 Through 20-06-2007.

\bibitem[\protect\citeauthoryear{Oyama, Hara, and Akamatsu}{Oyama
  et~al.}{2020}]{Oyama2020}
\textsc{Oyama, Y., Y.~Hara, and T.~Akamatsu} (2020): \enquote{{Markovian
  Traffic Equilibrium Assignment based on Network Generalized Extreme Value
  Model},} \emph{arXiv}.

\bibitem[\protect\citeauthoryear{Prato}{Prato}{2009}]{Prato2009}
\textsc{Prato, C.~G.} (2009): \enquote{{Route choice modeling: Past, present
  and future research directions},} \emph{Journal of Choice Modelling}, 2,
  65--100.

\bibitem[\protect\citeauthoryear{Rasmussen, Nielsen, Watling, and
  Giacomo~Prato}{Rasmussen et~al.}{2017}]{Rasmussen2017}
\textsc{Rasmussen, T.~K., O.~A. Nielsen, D.~P. Watling, and C.~Giacomo~Prato}
  (2017): \emph{European Journal of Transport and Infrastructure Research}, 17.

\bibitem[\protect\citeauthoryear{Rockafellar}{Rockafellar}{1970}]{Rockafellar1970}
\textsc{Rockafellar, R.~T.} (1970): \emph{{Convex Analysis}}, Princeton, N.J.:
  Princeton University Press.

\bibitem[\protect\citeauthoryear{Shen, Zhang, and Akamatsu}{Shen
  et~al.}{1996}]{Shen1996a}
\textsc{Shen, W., H.~M. Zhang, and T.~Akamatsu} (1996): \enquote{{Cyclic flows,
  Markov process and stochastic traffic assignment},} \emph{Transportation
  Research Record}, 30, 1--34.

\bibitem[\protect\citeauthoryear{S{\o}rensen and Fosgerau}{S{\o}rensen and
  Fosgerau}{2020}]{Sorensen2019}
\textsc{S{\o}rensen, J. R.-V. and M.~Fosgerau} (2020): \enquote{{How McFadden
  met Rockafellar and learned to do more with less},} .

\bibitem[\protect\citeauthoryear{Waltz, Morales, Nocedal, and Orban}{Waltz
  et~al.}{2006}]{Waltz2006}
\textsc{Waltz, R.~A., J.~L. Morales, J.~Nocedal, and D.~Orban} (2006):
  \enquote{{An interior algorithm for nonlinear optimization that combines line
  search and trust region steps},} \emph{Mathematical Programming}, 107,
  391--408.

\bibitem[\protect\citeauthoryear{Watling, Rasmussen, Prato, and
  Nielsen}{Watling et~al.}{2015}]{Watling2015}
\textsc{Watling, D.~P., T.~K. Rasmussen, C.~G. Prato, and O.~A. Nielsen}
  (2015): \enquote{{Stochastic user equilibrium with equilibrated choice sets:
  Part I - Model formulations under alternative distributions and
  restrictions},} \emph{Transportation Research Part B: Methodological}, 77,
  166--181.

\bibitem[\protect\citeauthoryear{Watling, Rasmussen, Prato, and
  Nielsen}{Watling et~al.}{2018}]{Watling2018}
---\hspace{-.1pt}---\hspace{-.1pt}--- (2018): \enquote{{Stochastic user
  equilibrium with a bounded choice model},} \emph{Transportation Research Part
  B: Methodological}, 114, 254--280.

\end{thebibliography}

\end{document}